\documentclass[sigconf, noncam]{acmart}

\newcommand\vldbdoi{XX.XX/XXX.XX}
\newcommand\vldbpages{XXX-XXX}
\newcommand\vldbvolume{16}
\newcommand\vldbissue{3}
\newcommand\vldbyear{2022}
\newcommand\vldbauthors{\authors}
\newcommand\vldbtitle{\shorttitle} 

\newcommand\vldbpagestyle{empty} 

\usepackage{balance}  
\usepackage{graphicx}
\usepackage{subcaption}
\usepackage{array}
\newcolumntype{M}[1]{>{\arraybackslash}m{#1}}
\usepackage{comment}
\usepackage{amsmath,amsthm}
\usepackage{url}
\usepackage{times,epsfig}
\usepackage{savesym}
\usepackage{verbatim}
\usepackage{multirow}
\usepackage{epstopdf}
\usepackage{caption}
\usepackage{subcaption}
\usepackage{lipsum}
\usepackage{setspace}
\usepackage{etoolbox}
\usepackage{tikz}
\usepackage{balance}
\usepackage{diagbox}
\usepackage{amsfonts}
\usetikzlibrary{automata, arrows, calc}

\usepackage{enumitem}
\usepackage{mathrsfs}
\usepackage[show]{chato-notes}
\usepackage{listings}
\usepackage{titlesec}
\usepackage{pifont}
\usepackage{xcolor} 
\definecolor{LightGray}{gray}{0.9}

\usepackage[linesnumbered,algoruled,boxed,lined,noend]{algorithm2e}
\usepackage{breqn}
 
\SetVlineSkip{0pt}

\SetCommentSty{mycommfont}

\newtheorem{definition}{Definition}[section]

\newtheorem{theorem}{Theorem}[section]
\newtheorem{lemma}[theorem]{Lemma}
\newtheorem{property}[theorem]{Property}

\newtheorem{example}{Example}[section]
\setcopyright{none}
\renewcommand\footnotetextcopyrightpermission[1]{}
\begin{document}

\newcommand{\cal}[1]{\mathcal{#1}}
\newcommand{\byronnotes}[1]{\textcolor{blue}{\noindent Note: #1}}
\newcommand{\choi}[1]{\textcolor{blue}{#1}}
\newcommand{\jiaxin}[1]{{#1}}
\newcommand{\TODO}[1]{\textcolor{BurntOrange}{\textbf{TODO:} #1}}
\newcommand{\Remind}[1]{\textcolor{RubineRed}{\textbf{Reminder:}#1}}
\newcommand{\byronsuggestion}[1]{\textcolor{Green}{\textbf{Reminder:}#1}}
\newcommand{\red}[1]{\textcolor{red}{#1}}
\newcommand{\what}[1]{\textcolor{blue}{?#1?}}
\newcommand{\eat}[1]{}
\newcommand{\tr}[1]{}
\newcommand{\kw}[1]{{\ensuremath {\textsf{#1}}}\xspace}
\newenvironment{tbi}{\begin{itemize}
		\setlength{\topsep}{0.6ex}\setlength{\itemsep}{0ex}} 
	{\end{itemize}} 
\newcommand{\ei}{\end{itemize}}

\newcommand{\Gontology}{G_{Ont}}
\newcommand{\PRADS}{\textsf{\small PADS}}
\newcommand{\KPADS}{\textsf{\small KPADS}}
\newcommand{\BPADS}{\textsf{\small BPADS}}
\newcommand{\ADS}{\textsf{\small ADS}}
\newcommand{\PKD}{\textsf{\small PKD}}
\newcommand{\FRAMEWORK}{\kw{FRAMEWORK\_NAME}}
\newcommand{\ALGO}{\kw{ALGO\_NAME}}
\newcommand{\PEval}{\kw{PEval}}
\newcommand{\IncEval}{\kw{IncEval}}
\newcommand{\Assemble}{\kw{Assemble}}
\newcommand{\ARef}{\textsf{\small ARefine}}
\newcommand{\ACmpl}{\textsf{\small AComplete}}
\newcommand{\DataGraph}{$G$}
\newcommand{\Ghier}{\mathbb{G}}
\newcommand{\qhier}{\mathbb{Q}}
\newcommand{\subclassOf}{\mathsf{SubClassOf}}
\newcommand{\subtypeOf}{\mathsf{SubTypeOf}}
\newcommand{\answer}{\mathsf{ans}}
\newcommand{\answerb}{\mathsf{ans}^b}
\newcommand{\answerf}{\mathsf{ans}^f}
\newcommand{\Bisim}{\mathsf{Bisim}}
\newcommand{\rank}{\mathsf{rank}}
\newcommand{\dist}{\mathsf{dist}}
\newcommand{\Next}{\mathsf{next}}
\newcommand{\distance}{\mathsf{d}}
\newcommand{\reach}{\mathsf{reach}}
\newcommand{\Generalization}{\mathsf{Gen}}
\newcommand{\Specialization}{\mathsf{Spec}}
\newcommand{\desc}{\mathsf{desc}}
\newcommand{\support}{\mathsf{sup}}
\newcommand{\distort}{\mathsf{DT}}
\newcommand{\degree}{\mathsf{deg}}
\renewcommand{\equiv}{\mathsf{equiv}}
\newcommand{\equivv}[1]{[#1]_\mathsf{equiv}}
\newcommand{\irchi}[2]{\raisebox{\depth}{$#1\chi$}}
\newcommand{\Summarization}{\mathpalette\irchi\relax}
\newcommand{\Configuration}{C}
\newcommand{\radius}{d_{max}}
\newcommand{\Eval}{\mathsf{eval}}
\newcommand{\peval}{\mathsf{peval}}
\newcommand{\F}{{\cal F}}
\newcommand{\Score}{\mathsf{scr}}
\newcommand{\ie}{\emph{i.e.,}\xspace}
\newcommand{\eg}{\emph{e.g.,}\xspace}
\newcommand{\wrt}{\emph{w.r.t.}\xspace}
\newcommand{\aka}{\emph{a.k.a.}\xspace}

\newcommand{\equi}{\mathsf{equi}}
\newcommand{\sn}{\mathsf{sn}}
\newcommand{\METIS}{\mathsf{\small METIS}}

\newcommand{\maxsf}{\mathsf{max}}
\newcommand{\PIE}{\mathsf{PIE}}
\newcommand{\PINE}{\mathsf{PINE}}

\newcommand{\threshold}{\tau}

\newcommand{\cost}{\mathsf{cost}}
\newcommand{\costq}{\mathsf{cost}_\mathsf{q}}
\newcommand{\CR}{\mathsf{compress}}
\newcommand{\DT}{\mathsf{distort}}
\newcommand{\FP}{\mathsf{fp}}
\newcommand{\maxSAT}{\mathsf{maxSAT}}
\newcommand{\True}{\mathsf{T}}
\newcommand{\False}{\mathsf{F}}
\newcommand{\OptGen}{\mathsf{OptGen}}
\newcommand{\freq}{\mathsf{freq}}

\newcommand{\match}{\mathsf{match}}

\newcommand{\vpd}{V_{pd}}
\newcommand{\vtp}{V_{tbp}}
\newcommand{\BiGindex}{\mathsf{BiG\textnormal{-}index}}
\newcommand{\filter}{\mathsf{filter}}
\newcommand{\ans}{\mathsf{ans\_graph\_gen}}
\newcommand{\Azero}{\mathbb{A}}
\newcommand{\boost}{\mathsf{boost}}
\newcommand{\bkws}{\mathsf{bkws}}
\newcommand{\fkws}{\mathsf{fkws}}
\newcommand{\rkws}{\mathsf{rkws}}
\newcommand{\dkws}{\mathsf{dkws}}
\newcommand{\knk}{\mathsf{knk}}
\newcommand{\ksp}{\mathsf{ksp}}
\newcommand{\config}{\mathsf{config}}
\newcommand{\content}{\mathsf{isKey}}
\newcommand{\pcnt}{\mathsf{pcnt}}
\newcommand{\private}{\textsf{isPrivate}}

\newcommand{\Path}{{\mathcal{P}}}
\renewcommand{\P}{{\mathcal P}}

\newcommand{\ppkws}{\textsf{\small PPKWS}\xspace}

\newcommand{\DKWS}{\textsf{\small DKWS}\xspace}
\newcommand{\kDKWS}{\textsf{\small $k$DKWS}\xspace}
\newcommand{\SKWS}{\textsf{\small BFKWS}\xspace}
\newcommand{\KWS}{\textsf{\small KWS}\xspace}
\newcommand{\VU}{\mathbb{V}}
\newcommand{\VI}{\mathcal{V}}
\newcommand{\VM}{\bar{\mathcal{V}}}
\newcommand{\vsf}{\mathsf{v}}
\newcommand{\usf}{\mathsf{u}}
\newcommand{\answerset}{\mathcal{A}}
\newcommand{\candanswerset}{\bar{\mathcal{A}}}
\newcommand{\prune}{S}
\newcommand{\Ud}{\hat{\dist}}
\newcommand{\Ld}{\check{\dist}}
\newcommand{\invert}{\mathscr{I}}
\newcommand{\MB}{u.b}
\newcommand{\MF}{\mathsf{f}}
\newcommand{\Queue}{\mathcal{P}}
\newcommand{\Visit}{\mathsf{Vis}}
\newcommand{\invertV}{V_{\invert}}
\newcommand{\invertE}{E_{\invert}}
\newcommand{\tnormal}[1]{\textnormal{#1}}

\newcommand{\DKWSBF}{\textsf{\small BF}\xspace}
\newcommand{\DKWSNP}{\textsf{\small BF+PADS+NP}\xspace}
\newcommand{\DKWSPADS}{\textsf{\small BF+PADS}\xspace}
\newcommand{\DKWSPINE}{\textsf{\small BF+ALL}\xspace}
\newcommand{\Notify}{\mathsf{Notify}}
\newcommand{\Push}{\mathsf{Push}}
\newcommand{\Parameters}{\mathscr{X}}
\newcommand{\grape}{\kw{GRAPE}}
\newcommand{\Buffer}{\mathbb{B}}
\newcommand{\SI}{\kw{SI}}

\newcommand{\SBGindex}{\mathsf{SBGIndex}}

\newcommand{\SGIndex}{\mathsf{SGIndex}} 

\newcommand{\Portal}{\mathbb{P}}
\newcommand{\rclique}{\mathsf{r\textnormal{-}clique}}
\newcommand{\Blinks}{\mathsf{Blinks}}
\newcommand{\Rclique}{\mathsf{Rclique}}

\newcommand{\pprclique}{\mathsf{PP\textnormal{-}r\textnormal{-}clique}}
\newcommand{\ppknk}{\mathsf{PP\textnormal{-}knk}}
\newcommand{\ppBlinks}{\mathsf{PP\textnormal{-}Blinks}}

\newcommand{\baselinerclique}{\mathsf{Baseline\textnormal{-}r\textnormal{-}clique}}
\newcommand{\baselineknk}{\mathsf{Baseline\textnormal{-}knk}}
\newcommand{\baselineBlinks}{\mathsf{Baseline\textnormal{-}Blinks}}

\newcommand{\stitle}[1]{\vspace{0.4ex}\noindent{\bf #1}}
\newcommand{\etitle}[1]{\vspace{0.8ex}\noindent{\underline{\em #1}}}
\newcommand{\eetitle}[1]{\vspace{0.6ex}\noindent{{\em #1}}}

%
\newcommand{\techreport}[2]{#2}
\newcommand{\SGFrame}{\mathsf{SGFrame}} 

\newcommand{\stab}{\rule{0pt}{8pt}\\[-2.0ex]}
\newcommand{\tab}{\hspace{4ex}}

\newcommand{\Q}{{\cal Q}}



\newcommand{\eop}{\hspace*{\fill}\mbox{\qed}}








\newcommand{\Src}{S}
\newcommand{\Dst}{T}
\newcommand{\SD}{At least $k$ $\mathsf{S}$-$\mathsf{T}$ maximum-flow}
\newcommand{\SDMF}{$\mathsf{kSTMF}$}
\newcommand{\SDMFG}{$\mathsf{kSTMF}^g$}
\newcommand{\TEMSDMF}{$\mathsf{TEM}$-$\mathsf{kSTMF}$}
\newcommand{\STcore}{$\mathsf{ST}$-FCore}
\newcommand{\Core}{\mathsf{Core}}
\newcommand{\MFlow}{\mathsf{MaxFlow}}
\newcommand{\MFavg}{MF}
\newcommand{\DKS}{\mathsf{DkS}}
\newcommand{\CBB}{C^T}
\newcommand{\fT}{f^T}
\newcommand{\tsf}{\mathsf{t}}
\newcommand{\csf}{\mathsf{c}}
\newcommand{\fsf}{\mathsf{f}}
\newcommand{\Tsf}{\mathsf{T}_{max}}
\newcommand{\algo}{\mathsf{algo}}
\newcommand{\FnDense}{\mathsf{FnDense}}
\newcommand{\FnSparse}{\mathsf{FnSparse}}
\newcommand{\Transform}{\hat{G}}
\newcommand{\TFNet}{$\mathsf{TF}$-$\mathsf{Network}$}
\newcommand{\RTFNet}{$\mathsf{RTF}$-$\mathsf{Network}$}
\newcommand{\Baseline}{\mathsf{Baseline}}
\newcommand{\Greedy}{\mathsf{Greedy}}
\newcommand{\BWCC}{$\mathsf{WCC}$}
\newcommand{\Tran}{\mathsf{T}}
\newcommand{\ts}{\tau}
\newcommand{\name}{\red{\mathsf{STMflow}}}
\newcommand{\FF}{$\mathsf{Flow}$-$\mathsf{Force}$}


\newcommand{\Spade}{$\mathsf{Spade}$}

\newcommand{\Seq}{O}
\newcommand{\Grab}{\mathsf{Grab}}
\newcommand{\DENG}{\mathsf{DG}}
\newcommand{\DENGW}{\mathsf{DW}}
\newcommand{\Fraudar}{\mathsf{FD}}
\newcommand{\IncDENG}{$\mathsf{IncDG}$}
\newcommand{\IncDENGW}{$\mathsf{IncDW}$}
\newcommand{\IncFraudar}{$\mathsf{IncFD}$}
\newcommand{\IncDENGU}{$\mathsf{IncDGG}$}
\newcommand{\IncDENGWU}{$\mathsf{IncDWG}$}
\newcommand{\IncFraudarU}{$\mathsf{IncFDG}$}
\newcommand{\permutation}{\alpha}
\newcommand{\TCal}{$\mathcal{T}$}
\newcommand{\TCald}{$\mathcal{T}^d$}
\newcommand{\AFF}{$G_{\mathcal{T}}$}
\newcommand{\latency}{\mathcal{L}}
\newcommand{\Elapsed}{\mathcal{E}}
\newcommand{\Ratio}{\mathcal{R}}

\newcommand{\DG}{$\mathbf{DG}$}
\newcommand{\DW}{$\mathbf{DW}$}
\newcommand{\FD}{$\mathbf{FD}$}

\newtheorem{manualtheoreminner}{Theorem}
\newenvironment{manualtheorem}[1]{%
  \renewcommand\themanualtheoreminner{#1}%
  \manualtheoreminner
}{\endmanualtheoreminner}

\newtheorem{manuallemmainner}{Lemma}
\newenvironment{manuallemma}[1]{%
  \renewcommand\themanuallemmainner{#1}%
  \manuallemmainner
}{\endmanualtheoreminner}

\newcommand*\circled[1]{\tikz[baseline=(char.base)]{
            \node[shape=circle,draw,inner sep=2pt] (char) {#1};}}



\title{Spade: A Real-Time Fraud Detection Framework on Evolving Graphs (Complete Version)}

\author{Jiaxin Jiang}
\affiliation{
  \institution{National University of Singapore}
}
\email{jxjiang@nus.edu.sg}

\author{Yuan Li}
\affiliation{
  \institution{National University of Singapore}
}
\email{li.yuan@u.nus.edu}

\author{Bingsheng He}
\affiliation{
  \institution{National University of Singapore}
}
\email{hebs@comp.nus.edu.sg}

\author{Bryan Hooi}
\affiliation{
  \institution{National University of Singapore}
}
\email{bhooi@comp.nus.edu.sg}

\author{Jia Chen}
\affiliation{
  \institution{GrabTaxi Holdings}
}
\email{jia.chen@grab.com}

\author{Johan Kok Zhi Kang}
\affiliation{
  \institution{GrabTaxi Holdings}
}
\email{johan.kok@grabtaxi.com}

\begin{abstract}
Real-time fraud detection is a challenge for most financial and electronic commercial platforms. To identify fraudulent communities, $\Grab$, one of the largest technology companies in Southeast Asia, forms a graph from a set of transactions and detects dense subgraphs arising from abnormally large numbers of connections among fraudsters. Existing dense subgraph detection approaches focus on static graphs without considering the fact that transaction graphs are highly dynamic. Moreover, detecting dense subgraphs from scratch with graph updates is time consuming and cannot meet the real-time requirement in industry. To address this problem, we introduce an incremental real-time fraud detection framework called \Spade{}\eat{which is deployed in $\Grab$}. \Spade{} can detect fraudulent communities in hundreds of microseconds on million-scale graphs by incrementally maintaining dense subgraphs. Furthermore, \Spade{} supports batch updates and 
edge grouping to reduce response latency. Lastly, \Spade{} provides simple but expressive APIs for the design of evolving fraud detection semantics. Developers plug their customized suspiciousness functions into \Spade{} which incrementalizes their semantics without recasting their algorithms. Extensive experiments show that \Spade{} detects fraudulent communities in real time on million-scale graphs. Peeling algorithms incrementalized by \Spade{} are up to a million times faster than the static version.
\end{abstract}


\maketitle

\pagestyle{\vldbpagestyle}
\begingroup\small\noindent\raggedright\textbf{PVLDB Reference Format:}\\
\vldbauthors. \vldbtitle. PVLDB, \vldbvolume(\vldbissue): \vldbpages, \vldbyear.\\
\href{https://doi.org/\vldbdoi}{doi:\vldbdoi}
\endgroup
\begingroup
\renewcommand\thefootnote{}\footnote{\noindent
This work is licensed under the Creative Commons BY-NC-ND 4.0 International License. Visit \url{https://creativecommons.org/licenses/by-nc-nd/4.0/} to view a copy of this license. For any use beyond those covered by this license, obtain permission by emailing \href{mailto:info@vldb.org}{info@vldb.org}. Copyright is held by the owner/author(s). Publication rights licensed to the VLDB Endowment. \\
\raggedright Proceedings of the VLDB Endowment, Vol. \vldbvolume, No. \vldbissue\ %
ISSN 2150-8097. \\
\href{https://doi.org/\vldbdoi}{doi:\vldbdoi} \\
}\addtocounter{footnote}{-1}\endgroup

\section{Introduction}\label{sec:intro}

Graphs have been found in many emerging applications, including transaction networks, communication networks and social networks. The dense subgraph problem is first studied in~\cite{goldberg1984finding} and is effective for link spam identification~\cite{gibson2005discovering,beutel2013copycatch}, community detection~\cite{dourisboure2007extraction,chen2010dense} and fraud detection~\cite{hooi2016fraudar,chekuri2022densest,shin2016corescope}. Standard peeling algorithms~\cite{tsourakakis2015k,hooi2016fraudar,bahmani2012densest,chekuri2022densest,boob2020flowless} iteratively peel the vertex that has the smallest connectivity (\eg vertex degree or sum of the weights of the adjacent edges) to the graph. Peeling algorithms are widely used because of their efficiency, robustness, and theoretical worst-case guarantee. However, existing peeling algorithms~\cite{hooi2016fraudar,tsourakakis2015k,charikar2000greedy} assume a static graph without considering the fact that social and transaction graphs in online marketplaces are rapidly evolving in recent years. One possible solution for fraud detection on evolving graphs is to perform peeling algorithms periodically. We take $\Grab$'s fraud detection pipeline as an example.

\eat{There have been quite a lot of techniques \eg \cite{hooi2016fraudar,beutel2013copycatch,shin2016corescope,kumar1999trawling,ren2021ensemfdet,khuller2009finding, tsourakakis2013denser,ma2020efficient} proposed for detecting dense subgraphs.}

\begin{figure}[tb]
    \includegraphics[width=0.7\linewidth]{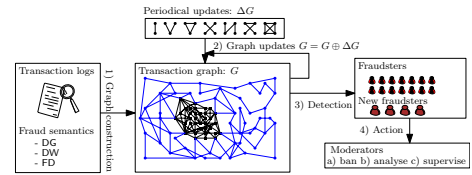}
    \caption{$\mathsf{\mathbf{Grab}}$'s data pipeline for fraud detection}\label{fig:pipeline}
\end{figure}

\begin{figure}[tb]
    \includegraphics[width=0.8\linewidth]{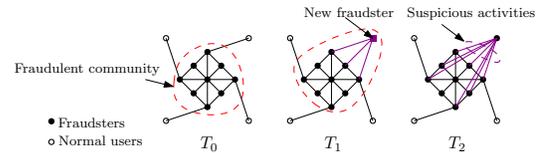}
    \caption{An example of fraud detection on evolving graphs}\label{fig:intro}
\end{figure}

\eat{However, existing works~\cite{hooi2016fraudar,tsourakakis2015k,charikar2000greedy} are proposed for static graphs. When graphs change, we have to detect dense subgraphs from scratch, which cannot meet the real-time requirement.}

\eat{Our experiments show that hundreds of new transactions are generated each second on average in 2021.}

\eat{
\stitle{Evolving graphs}.   Although peeling algorithms are efficient, identifying the dense subgraph still takes several minutes. Therefore, given graph updates, the detection of dense subgraphs from scratch cannot meet the real-time requirement. 
}

\stitle{Fraud detection pipeline in $\Grab$ (Figure~\ref{fig:pipeline})}. $\Grab$ is one of the largest technology companies in Southeast Asia and offers digital payments and food delivery services. On the $\Grab$'s e-commerce platform, 1) the transactions form a transaction graph $G$. 2) $\Grab$ updates the transaction graphs periodically $G = G\oplus \Delta G$. Our experiments show that it takes $28$s to carry out $\mathsf{Fraudar}$ ($\Fraudar{}$)~\cite{hooi2016fraudar} on a transaction graph with $6$M vertices and $25$M edges. Therefore, we can execute fraud detection every $30$ seconds. 3) The dense subgraph detection algorithm and its variants are used to detect fraudulent communities. 4) After identifying the fraudsters, the moderators ban or freeze their accounts to avoid further economic loss. A classic fraud example is customer-merchant collusion. Assume that $\Grab$ provides promotions to new customers and merchants. However, fraudsters create a set of fake accounts and do fictitious trading to use the opportunity of promotion activities to earn the bonus. Such fake accounts (vertices) and the transactions among them (edges) form a dense subgraph.

\setlist{nolistsep}
\begin{example}
    Consider the transaction graph in Figure~\ref{fig:intro}, where a vertex is a user or a store, and an edge represents a transaction. Suppose a fraudulent community is identified at time $T_0$ and a normal user becomes a fraudster and participates in suspicious activities at $T_1$. Applying peeling algorithms at $T_1$, the new fraudster is detected at $T_2$. However, many new suspicious activities have occurred during the time period $[T_1,T_2]$ that could cause huge economic losses.
\end{example}

As reported in recent studies~\cite{dailyreport,ye2021gpu}, $21.4\%$ of the traffic to e-commerce portals are malicious bots in 2018. Fraud detection is challenging since many fraudulent activities occur in a very short timespan. Hence, identifying fraudsters and reducing response latency to fraudulent transactions are key tasks in real-time fraud detection.

\begin{table}[tb]
        \caption{Comparison of \Spade{} and previous algorithms}\label{table:comparison}
    \centering
   \begin{scriptsize}
    \begin{tabular}{|c|c|c|c|c|c|}
      \hline
       & $\DENG$~\cite{charikar2000greedy} & $\DENGW$~\cite{gudapati2021search}  &  $\Fraudar$~\cite{hooi2016fraudar} & \Spade  \\
      \hline
      Dense subgraph detection & \checkmark & \checkmark  & \checkmark & \checkmark    \\
      \hline
      Accuracy guarantees & \checkmark & \checkmark & \checkmark & \checkmark \\
      \hline
      Weighted graph & \ding{55} & \checkmark &   \checkmark & \checkmark \\  
      \hline
      Incremental updates & \ding{55} & \ding{55} & \ding{55} & \checkmark \\
        \hline
      Edge reordering & \ding{55} & \ding{55} & \ding{55} & \checkmark\\
      \hline
      
    \end{tabular}
    \end{scriptsize}
\end{table}

To address real-time fraud detection on evolving graphs, a better solution would be to incrementally maintain dense subgraphs. There are two main challenges of incremental maintenance. First, \eat{real-time fraud detection is necessary as fraud detection is an essential path to respond to customer's requests.} operational demands require that fraudsters should be identified in $100$ milliseconds in industry. Maintaining the dense subgraph incrementally in such a short timespan is challenging. Second, fraud semantics continue to evolve and it is not trivial to incrementalize each of them. Implementing a correct and efficient incremental algorithm is, in general, a challenge. It is impractical to train all developers with the knowledge of incremental graph evaluation. To the best of our knowledge, there are no generic approaches to minimize the cost of incremental peeling algorithms. Motivated by the challenges, we design a real-time fraud detection framework, named \Spade{} to detect fraudulent communities by incrementally maintaining dense subgraphs. The comparison between \Spade{} and the previous algorithms (\jiaxin{dense subgraphs ($\DENG$)~\cite{charikar2000greedy}, dense weighted subgraph ($\DENGW$)~\cite{gudapati2021search} and $\mathsf{Fraudar}$ ($\Fraudar$)~\cite{hooi2016fraudar}}) is summarized in Table~\ref{table:comparison}.

\stitle{Contributions.} In this paper, we focus on incremental peeling algorithms. In summary, this paper makes the following contributions.

\begin{enumerate}
    \item We build three fundamental incremental techniques for peeling algorithms to avoid detecting fraudulent communities from scratch. \Spade{} inspects the subgraph that is affected by graph updates and reorders the peeling sequence incrementally, which theoretically guarantees the accuracy of the worst case.
    \item \Spade{} enables developers to design their fraud semantics to detect fraudulent communities by providing the suspiciousness functions of edges and vertices. We show that a variety of peeling algorithms can be incrementalized in \Spade{} (Section~\ref{sec:framework}) including $\DENG$, $\DENGW$ and $\Fraudar$.
    \item We conduct extensive experiments on \Spade{} with datasets from industry. The results show that \Spade{} speeds up fraud detection up to $6$ orders of magnitude since \Spade{} minimizes the cost of incremental maintenance by inspecting the affected area. Furthermore, the latency of the response to fraud activities can be significantly reduced. Lastly, once a user is spotted as a fraudster, we identify the related transactions as potential fraud transactions and pass them to system moderators. Up to $88.34\%$ potential fraud transactions can be prevented. 
\end{enumerate}

\stitle{Organization.} The rest of this paper is organized as follows: Section~\ref{sec:background} presents the background and the problem statement. We introduce the framework of \Spade{} in Section~\ref{sec:framework} and three incremental peeling algorithms in Section~\ref{sec:Spade}. Section~\ref{sec:exp} reports on the experimental evaluation. After reviewing related work in Section~\ref{sec:related}, we conclude in Section~\ref{sec:conclusion}.

\section{Background}\label{sec:background}

\subsection{Preliminary}

\eat{
We next introduce some basic notations. Some frequently used notations are summarized in Table~\ref{tab-notations}.
}

\begin{table}[tb!]
	  \caption{Frequently used notations}
	  \setlength{\tabcolsep}{0.5em}
	\label{tab-notations}
	\begin{scriptsize}
		\begin{center}
			\begin{tabular}
				{|c|c|} \hline Notation & Meaning \\
				\hline
                                \hline
                $G$ / $\Delta G$ & a transaction graph / updates to graph $G$ \\ \hline
				$G\oplus \Delta G$ & the graph obtained by updating $\Delta G$ to $G$ \\ \hline
				$a_i$ / $c_{ij}$ & the weight on vertex $u_i$ / on edge $(u_i,u_j)$ \\ \hline
				$f(S)$ & the sum of the suspiciousness of induced subgraph $G[S]$ \\ \hline
				$g(S)$ & the suspiciousness density of  vertex set $S$ \\ \hline
				$w_{u}(S)$ & peeling weight, \ie the decrease in $f$ by removing $u$ from $S$ \\ \hline
				$Q$ & a peeling algorithm\\ \hline
				$\Seq$ & the peeling sequence order \wrt $Q$\\ \hline
				$S^P$ & the vertex set returned by a peeling algorithm \\ \hline
				$S^*$ & the optimal vertex set, \ie $g(S^*)$ is maximized \\  \hline
			\end{tabular}
		\end{center}
	\end{scriptsize}
\end{table}

\stitle{Graph $G$.} We consider a directed and weighted graph $G=(V,E)$, where $V$ is a set of vertices and $E$ $\subseteq (V\times V)$ is a set of edges. Each edge $(u_i,u_j)\in E$ has a \textbf{nonnegative} weight, denoted by $c_{ij}$. We use $N(u)$ to denote the neighbors of $u$.

\stitle{Induced subgraph.} Given a subset $S$ of $V$, we denote the induced subgraph by $G[S] = (S, E[S])$, where $E[S]=\{(u,v) | (u,v) \in E \wedge u,v\in S\}$. We denote the size of $S$ by $|S|$.

\stitle{Density metrics $g$.} We adopt the class of metrics $g$ in previous studies~\cite{hooi2016fraudar,gudapati2021search,charikar2000greedy}, $g(S) = \frac{f(S)}{|S|}$, where $f$ is the total weight of $G[S]$, \ie the sum of the weight of $S$ and $E[S]$:

\begin{equation}\label{eq:density}
    f(S)=\sum_{u_i\in S} a_i + \sum_{u_i,u_j\in S \bigwedge (u_i,u_j)\in E} c_{ij}
\end{equation}

The weight of a vertex $u_i$ measures the suspiciousness of user $u_i$, denoted by $a_i$ ($a_i\geq 0$). The weight of the edge $(u_i,u_j)$ measures the suspiciousness of transaction $(u_i,u_j)$, denoted by $c_{ij} > 0$. Intuitively, $g(S)$ is the density of the induced subgraph $G[S]$. The larger $g(S)$ is, the denser $G[S]$ is.

\stitle{Graph updates $\Delta G$.} We denote the set of updates to $G$ by $\Delta G = (\Delta V, \Delta E)$. We denote the graph obtained by updating $\Delta G$ to $G$ as $G\oplus \Delta G$. Since transaction graphs continue to evolve, we consider edge insertion rather than edge deletion. Therefore, $G\oplus \Delta G = (V\cup \Delta V, E\cup \Delta E)$. Specifically, we consider two types of updates, edge insertion (\ie $|\Delta E| = 1$) and edge insertion in batch (\ie $|\Delta E| > 1$).

\subsection{Peeling algorithms}

\begin{algorithm}[tb]
    \caption{Execution paradigm of peeling algorithms}\label{algo:peeling}
    \footnotesize
    \SetKwProg{Fn}{Function}{}{}
    \KwIn{A graph $G = (V, E)$ and a density metric $g(S)$}
    \KwOut{The peeling sequence order $\Seq = Q(G)$ and the fraudulent community}

    $S_0 = V$ \label{algo:peeling:init} \\

    \For{$i=1,\ldots, |V|$}{
        select the vertex $u\in S_{i-1}$ such that $g(S_{i-1}\setminus \{u\})$ is maximized \label{algo:peeling:maximize}\\
        $S_{i} = S_{i-1}\setminus \{u\}$ \label{algo:peeling:maximize2} \\ 
        $\Seq.\mathsf{add}$($u$) 
    }
    \Return{$\Seq$ \textnormal{and} $\arg\max_{S_i}g(S_i)$}
\end{algorithm}

Peeling algorithms ($Q$) are widely used in dense subgraph mining~\cite{hooi2016fraudar,tsourakakis2015k,charikar2000greedy}. They follow the execution paradigm in Algorithm~\ref{algo:peeling} and differ mainly in density metrics. They are categorized to three categories: unweighted~\cite{charikar2000greedy}, edge-weighted~\cite{gudapati2021search} and hybrid-weighted~\cite{hooi2016fraudar}.

\stitle{Peeling weight.} Specifically, we use $w_{u_i}(S)$ to indicate the decrease in the value of $f$ when the vertex $u_i$ is removed from a vertex set $S$, \ie the peeling weight. Previous work~\cite{hooi2016fraudar} formalizes $w_{u_i}(S)$ as follows:

\begin{equation}
    w_{u_i}(S) = a_i + \sum_{(u_j\in S) \bigwedge ((u_i,u_j)\in E)} c_{ij} + \sum_{(u_j\in S) \bigwedge ((u_j,u_i)\in E)} c_{ji} 
\end{equation}

\eat{where $c_{ij}$ is the weight of the edge $(u_i,u_j)$ indicating its suspiciousness, $a_i$ is the weight of the vertex $u_i$ indicating the suspiciousness of $u_i$, and $u_j\in N(u_i)$.}

\stitle{Peeling sequence}. We use $S_i$ to denote the vertex set after $i$-th peeling step. Initially, the peeling algorithms set $S_0 = V$. They iteratively remove a vertex $u_i$ from $S_{i-1}$, such that $g(S_{i-1}\setminus \{u_i\})$ is maximized (Line~\ref{algo:peeling:maximize}$\sim$\ref{algo:peeling:maximize2}). The process repeats recursively until there are no vertices left. This leads to a series of sets over $V$, denoted by $S_0, \ldots, S_{|V|}$ of sizes $|V|, \ldots, 0$. Then $S_i$ ($i\in [0, |V|]$), which maximizes the density metric $g(S_i)$, is returned, denoted by $S^P$.  For simplicity, we denote $\Delta_i = w_{u_i}(S_i)$. Instead of maintaining the series $S_0,\ldots, S_{|V|}$, we record the peeling sequence $\Seq = [u_1,\ldots u_{|V|}]$ such that $\{u_i\} = S_{i-1}\setminus S_{i}$.

\begin{example}
    Consider the graph $G$ in Figure~\ref{fig:peel}. $u_1$ is peeled since its peeling weight is the smallest among all vertices. Similarly, $u_3,u_2,u_4,u_5$ will be peeled accordingly. Therefore, the peeling sequence is $\Seq = [u_1,u_3,u_2,u_4,u_5]$.
\end{example}

\stitle{Complexity and accuracy guarantee.} In Algorithm~\ref{algo:peeling}, Min-Heap is used to maintain the peeling weights, the insertion cost is $O(\log |V|)$. There are at most $|E|$ insertions. Therefore, the complexity of Algorithm~\ref{algo:peeling} is $O(|E|\log |V|)$.  We denote the vertex set that maximizes $g$ by $S^*$. Previous studies~\cite{khuller2009finding,hooi2016fraudar,charikar2000greedy} conclude that:

\eat{, we extend the proof of the accuracy guarantee to general graphs as follows.}

\begin{lemma}\label{lemma:2ppr}
Let $S^P$ be the vertex set returned by the peeling algorithms and $S^*$ be the optimal vertex set, $g(S^P) \geq \frac{1}{2} g(S^*)$.
\end{lemma}

\eat{
\begin{proof}
    As visualized in Figure~\ref{fig:proof}, we assume that $u_i\in S^P$ is the first vertex in $S^*$ removed by the peeling algorithm. Since $u_i$ is the smallest of the decrease in the value of $f(S')$, we have 
    
    \begin{equation}\label{eq:proof}
    \footnotesize
        g(S') = \frac{f(S')}{|S'|} \geq  \frac{\sum\limits_{u_j\in S'} w_{u_j}(S')}{2\times |S'|} \geq \frac{w_{u_i}(S') |S'|}{2\times |S'|} = \frac{w_{u_i}(S')}{2}
    \end{equation}
    
Due to the definition of $S^P$, $g(S^P) \geq g(S')$. Combining Equation~\ref{eq:proof}, we have the following proof.
    
    \begin{equation}
    \footnotesize
        g(S^P) \geq g(S') = \frac{f(S')}{|S'|} \geq \frac{w_{u_i}(S')}{2} \geq \frac{w_{u_i}(S^*)}{2} \geq \frac{g(S^*)}{2}
    \end{equation}
\end{proof}
}

Although peeling algorithms are scalable and robust, we remark that these algorithms are proposed for static graphs, which takes several minutes on million-scale graphs. For evolving graphs, computing from scratch is still time-consuming, which cannot meet the real-time requirement. Moreover, it is not trivial to design incremental algorithms for peeling algorithms. In this paper, we investigate an auto-incrementalization framework for peeling algorithms.

\stitle{Problem definition.} Given a graph $G=(V,E)$, a peeling algorithm $Q$, and the peeling result of $Q$ on $G$, $S^P=Q(G)$, our problem is to efficiently identify the result of $Q$ on $G\oplus \Delta G$, ${S}^{P'}=Q(G\oplus \Delta G)$, where $\Delta G$ is the graph updates.

\begin{figure}
    \includegraphics[width=0.65\linewidth]{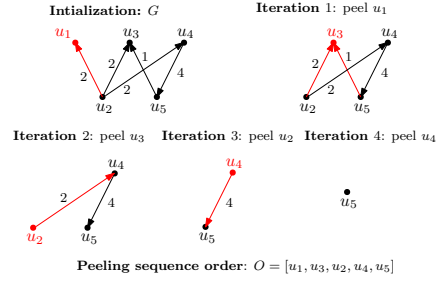}
    \caption{Example of peeling algorithms}\label{fig:peel}    
\end{figure}

\eat{
\begin{figure}
    \includegraphics[width=0.65\linewidth]{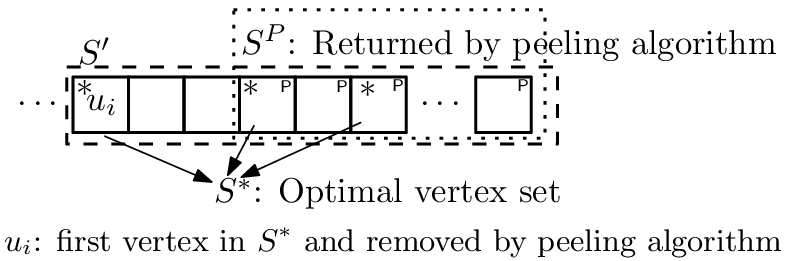}
    \caption{Illustration of key notations in peeling algorithms}\label{fig:proof}    
\end{figure}
}

\section{The \Spade{} Framework}\label{sec:framework}

In this section, we present an overview of our proposed framework \Spade{} and sample APIs. Subsequently, we demonstrate some examples on how to implement different peeling algorithms with \Spade.

\begin{figure}
\begin{minipage}[t]{.3\textwidth}
        \includegraphics[width=\linewidth]{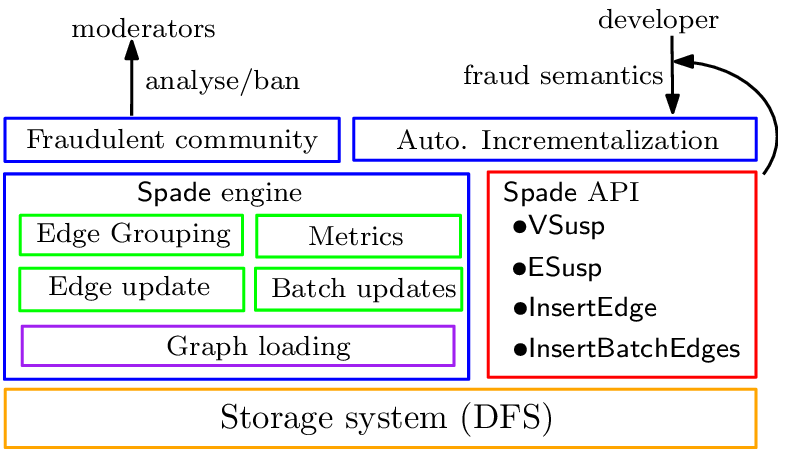}
        \subcaption[]{Architecture}
\end{minipage}
\hfill
\begin{minipage}[t]{.13\textwidth}
\includegraphics[width=\textwidth]{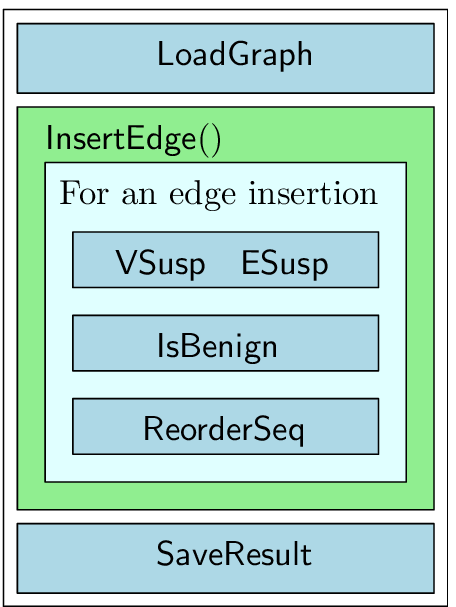}
\subcaption[]{Edge insertion
}
\end{minipage}
  \caption{Architecture of \Spade{} and workflow of an edge insertion}\label{fig:architecture}
\end{figure}

\subsection{Overview of \Spade{} and APIs}

\eat{
\begin{table*}[tb]
    \caption{Sample User-defined APIs in \Spade}\label{table:api}
    \vspace{-1em}
    \centering
    \begin{scriptsize}
    \begin{tabular}{|l|l|}
      \hline
     \textbf{APIs/Parameters} & \textbf{Descriptions} \\ \hline  
      $\mathsf{VSusp}$(Function $\mathsf{Score}$(Vertex $u$)) &  Given a suspiciousness score function $\mathsf{Score}$, \Spade{} plugins it in and computes the suspiciousness score for each vertex $u$\\
      \hline
      $\mathsf{ESusp}$(Function $\mathsf{Score}$(Edge $e$) &  Given a suspiciousness score function $\mathsf{Score}$, \Spade{} plugs it in and computes the suspiciousness for each edge $e$ \\
      \hline
      $\mathsf{InsertEdge}$(Edge $e$) & Given a new edge $e$, \Spade{} reorders the peeling sequence with the edge update \\
      \hline
      $\mathsf{InsertBatchEdges}$(Edge* \textit{e\_arr}) & Given a set of new transactions and the suspiciousness, \Spade{} reorders the peeling sequence with the batch updates \\
      \hline
    \end{tabular}
    \end{scriptsize}
    \end{table*}
}  

We follow two design goals to satisfy operational demands.

\begin{itemize}
    \item \etitle{Programmability.} We provide a set of user-defined APIs for developers to develop their dense subgraph-based semantics to detect fraudsters. Moreover, \Spade{} can auto-incrementalize their semantics without recasting the algorithms. 
    \item \etitle{Efficiency.} \Spade{} allow efficient and scalable fraud detection on evolving graphs in real-time.
\end{itemize}

\stitle{Architecture of \Spade.} Figure~\ref{fig:architecture} shows the architecture of \Spade{} and the workflow of an edge insertion. \Spade{} automatically incrementalizes peeling algorithms with the user-defined suspiciousness functions. To avoid computing from scratch on evolving graphs, the engine of \Spade{} maintains the fraudulent community incrementally with an edge update (Section~\ref{sec:edgebyedge}). Batch execution is developed to improve the efficiency of handling edge updates in batch (Section~\ref{sec:batch}). The updated fraudulent community is identified in real time and returned to the moderators for further analysis. Given an edge insertion, the workflow of \Spade{} contains the following components:

\begin{itemize}
    \item \underline{$\mathsf{VSusp}$ and $\mathsf{ESusp}$.} Given a new vertex/edge, these components are responsible for deciding the suspiciousness of the endpoint of the edge or the edge with a user-defined strategy.
    \item \underline{$\mathsf{IsBenign}$.} This component is responsible for deciding whether a new edge is benign (Section~\ref{sec:urgent}). If the edge is benign, it is inserted into an edge vector pending reordering; otherwise, peeling sequence reordering is triggered immediately for the edge buffer with this new edge.
    \item \underline{$\mathsf{ReorderSeq}$.} This component is responsible for incrementally maintaining the peeling sequence and deciding the new fraudulent community with the graph updates detailed in Section~\ref{sec:Spade}.
\end{itemize}

\definecolor{solarized@base03}{HTML}{002B36}
\definecolor{solarized@base02}{HTML}{073642}
\definecolor{solarized@base01}{HTML}{586e75}
\definecolor{solarized@base00}{HTML}{657b83}
\definecolor{solarized@base0}{HTML}{839496}
\definecolor{solarized@base1}{HTML}{93a1a1}
\definecolor{solarized@base2}{HTML}{EEE8D5}
\definecolor{solarized@base3}{HTML}{FDF6E3}
\definecolor{solarized@yellow}{HTML}{B58900}
\definecolor{solarized@orange}{HTML}{CB4B16}
\definecolor{solarized@red}{HTML}{DC322F}
\definecolor{solarized@magenta}{HTML}{D33682}
\definecolor{solarized@violet}{HTML}{6C71C4}
\definecolor{solarized@blue}{HTML}{268BD2}
\definecolor{solarized@cyan}{HTML}{2AA198}
\definecolor{solarized@green}{HTML}{859900}

\lstset{language=C++,
        basicstyle=\linespread{0.8}\footnotesize\ttfamily,
        numbers=left,
        numberstyle=\scriptsize,
        tabsize=1,
        breaklines=true,
        backgroundcolor=\color{white}, 
        escapeinside={@}{@},
        numberstyle=\tiny\color{solarized@base01},
        keywordstyle=\color{solarized@green},
        stringstyle=\color{solarized@cyan}\ttfamily,
        identifierstyle=\color{solarized@blue},
        commentstyle=\color{solarized@base01},
        emphstyle=\color{solarized@red},
        frame=single,
        rulecolor=\color{solarized@base2},
        rulesepcolor=\color{solarized@base2},
        showstringspaces=false,xleftmargin=0.2cm
}

\stitle{APIs and data structure (Listing~\ref{lst:listing-cpp}).} We provide APIs for developers to customize and deploy their peeling algorithms for different application requirements. Developers can customize $\mathsf{VSusp}$ and $\mathsf{ESusp}$ to develop their fraud detection semantics. We design two APIs for edge insertion, namely $\mathsf{InsertEdge}$ and $\mathsf{InsertBatchEdges}$. The $\mathsf{Detect}$ function spots the fraudulent community on the current graph. $\mathsf{IsBenign}$ and $\mathsf{ReorderSeq}$ are two built-in APIs which are transparent to developers. They are activated when new edges are inserted. \Spade{} uses the adjacency list to store the graph. Two vectors $\textit{_seq}$ and $\textit{_weight}$ are used to store the peeling sequence and the peeling weights. 

\begin{center}
  \scriptsize
\begin{lstlisting}[caption=Overview of \Spade, label={lst:listing-cpp}, language=C++]
class Spade {
 public:
  Graph LoadGraph(string path){} //Load graph from disk
  //Plug in vertex suspiciousness function 
  void VSusp(function<double(Vertex u, Graph g)> susp) {}
  //Plug in edge suspiciousness function 
  void ESusp(function<double(Edge e, Graph g)> susp) {}  
  //Detect the fraudsters on graph _g
  set<Vertex> Detect() {} 
  //Insert an edge and detect the new fraudsters
  set<Vertex> InsertEdge(Edge e) {}   
  //Insert a batch of edges and detect the new fraudsters
  set<Vertex> InsertBatchEdges(Edge* e_arr) {} 
 private:
  Graph _g; //Graph
  vector<Vertex> _seq; //Peeling sequence
  vector<double> _weight; //Peeling weights
  vector<Edge> _benign_edges; //Store the benign edges
  bool IsBenign(Edge e) {} //Judge if an edge is benign
  void ReorderSeq(){} //Reorder the peeling sequence 
}
\end{lstlisting}\label{apis}
\end{center}

\jiaxin{\stitle{Characteristic of density metrics.} We next formalize the sufficient condition of the density metrics that can be supported by \Spade{}.}

\begin{property}
If 1) $g(S)$ is an arithmetic density, \ie $g=\frac{|f(S)|}{|S|}$, 2) $a_i\geq 0$, and 3) $c_{ij} > 0$, then $g(S)$ is supported by \Spade{}.
\end{property}

\jiaxin{The correctness is satisfied since \Spade{} correctly returns the peeling sequence order (detailed in Section~\ref{sec:Spade}).  We also characterize the properties of these popular density metrics in Appendix~\ref{sec:axioms} of \cite{techreport}.}

\eat{
\noindent\jiaxin{\Spade{} focuses on the arithmetic density in the form of $g(S) = \frac{|f(S)|}{|S|}$ as introduced in Section~\ref{sec:background}. There are only two simple conditions of $f(S)$ (detailed in Equation~\ref{eq:density}) required by \Spade{}: i) for any $a_i$, the suspiciousness of $u_i$, is non-negative; and ii) for any $c_{ij}$, the suspiciousness of $(u_i,u_j)$ is positive.}
}

\tr{
\subsection{Instances of \Spade{}}
}

\stitle{Instances.} We show that popular peeling algorithms are easily implemented and supported by \Spade{}, \eg $\DENG$~\cite{charikar2000greedy}, $\DENGW{}$~\cite{gudapati2021search} and $\Fraudar{}$~\cite{hooi2016fraudar}. \jiaxin{We take $\Fraudar{}$ as an example and leave the discussion of the other instances in the Appendix~\ref{sec:instances} of~\cite{techreport}.} To resist the camouflage of fraudsters, Hooi et al. ~\cite{hooi2016fraudar} proposed $\Fraudar{}$ to weight edges and set the prior suspiciousness of each vertex with side information. Let $S\subseteq V$. The density metric of $\Fraudar$ is defined as follows:

\tr{
\etitle{Instance 1. Dense subgraphs ($\DENG{}$)~\cite{charikar2000greedy}.} $\DENG{}$ is designed to quantify the connectivity of substructures. It is widely used to detect fake comments~\cite{kumar2018community} and fraudulent activities~\cite{ban2018badlink} on social graphs. Let $S\subseteq V$. The density metric of $\DENG{}$ is defined by $g(S) = \frac{|E[S]|}{|S|}$. To implement $\DENG$ on \Spade{}, developers only need to design and plug in the suspiciousness function $\mathsf{esusp}$ by calling $\mathsf{ESusp}$. Specifically,  $\mathsf{esusp}$ is a constant function for edges, \ie $\mathsf{esusp}(u_i,u_j) = 1$.
}

\tr{
\etitle{Instance 2. Dense weighted subgraphs ($\DENGW$)~\cite{gudapati2021search}.} On transaction graphs, there are weights on the edges in usual, such as the transaction amount. The density metric of $\DENGW{}$ is defined by $g(S) = \frac{\sum_{(u_i,u_j)\in E[S]}c_{ij}}{|S|}$, where $c_{ij}$ is the weight of the edge $(u_i,u_j)\in E$. To implement $\DENGW{}$, users only need to plug in the suspiciousness function $\mathsf{esusp}$, \ie given an edge, $\mathsf{esusp}(u_i,u_j) = c_{ij}$.
}

\begin{equation}
    g(S) = \frac{f(S)}{|S|} = \frac{\sum_{u_i\in S} a_i + \sum_{u_i,u_j\in S \bigwedge (u_i,u_j)\in E} c_{i,j}}{|S|}
\end{equation}

\tr{
\begin{center}
  \scriptsize
\begin{lstlisting}[caption=Implementation of $\Fraudar{}$ on \Spade, label={lst:listing-fd}, language=C++]
double vsusp(Vertex v, Graph g){
  return g.weight[v]; //side information on vertex
}
double esusp(Edge e, Graph g){
  return 1/log(g.deg[e.src]+5); //user-defined function
}
int main() {
  Spade spade;  
  spade.VSusp(vsusp); //plug in vsusp (line 1-3) 
  spade.ESusp(esusp); //plug in esusp (line 4-6)
  spade.TurnOnEdgeGrouping(); //enable edge grouping
  spade.LoadGraph("graph_sample_path");
  vector<Vertex> fraudsters = spade.Detect();
  //edge insertions prepared by developers
  vector<Edge> edge_insertions;
  for(Edge e: edge_insertions){
      fraudsters = spade.InsertEdge(e);
  }
  return 0;
}
\end{lstlisting}
\end{center}
}

To implement $\Fraudar$ on \Spade{}, users only need to plug in the suspiciousness function $\mathsf{vsusp}$ for the vertices by calling $\mathsf{VSusp}$ and the suspiciousness function $\mathsf{esusp}$ for the edges by calling $\mathsf{ESusp}$. Specifically, 1) $\mathsf{vsusp}$ is a constant function, \ie given a vertex $u$, $\mathsf{vsusp}(u) = a_i$ and 2) $\mathsf{esusp}$ is a logarithmic function such that given an edge $(u_i,u_j)$, $\mathsf{esusp}(u_i,u_j) = \frac{1}{\log (x+c)}$, where $x$ is the degree of the object vertex between $u_i$ and $u_j$, and $c$ is a small positive constant ~\cite{hooi2016fraudar}.

Developers can easily implement customized peeling algorithms with \Spade{}, which significantly reduces the engineering effort. For example, users write only about $20$ lines of code (compared to about $100$ lines in the original $\Fraudar$~\cite{hooi2016fraudar}) to implement $\Fraudar$.

\section{Incremental peeling algorithms}\label{sec:Spade}

\eat{Given a set of updates to the graph, computing from scratch is costly. On the other hand, large suspicious transactions occur in a very short timeslot. Consequently, fraud cannot be identified and prevented at an early stage. Therefore, i} 

In this section, we propose several techniques to incrementally identify fraudsters by reordering the peeling sequence $\Seq$ with graph updates, \ie the peeling sequence on $G\oplus \Delta G$, denoted by $\Seq'$. 

\tr{In Section~\ref{sec:edgebyedge}, we introduce how to reorder the peeling sequence with edge insertion. In Section~\ref{sec:batch}, we elaborate on how to reorder the peeling sequence in batch. In Section~\ref{sec:urgent}, we show how to distinguish potential fraudulent transactions from benign transactions. \Spade{} groups the benign transaction insertion to reduce the latency of response to potential fraudulent transactions.}

\subsection{Peeling sequence reordering with edge insertion}\label{sec:edgebyedge}

Given a graph $G=(V,E)$, the peeling sequence $\Seq$ on $G$ and the graph updates $\Delta G=(\Delta V, \Delta E)$, where $|\Delta E| = 1$, \Spade{} returns the peeling sequence $\Seq'$ on $G\oplus \Delta G$.

\eat{
\begin{lemma}
    Given two vertices $u_i$ and $u_j$ where $i\leq j$, $w_{u_i}(S_i) \leq w_{u_j}(S_i)$.
\end{lemma}
\begin{proof}
    Since $i\leq j$, $S_i \subseteq S_j$. Hence $w_{u_i}(S_j)$
\end{proof}
}

\stitle{Vertex insertion.} Given a new vertex $u$, we insert it into the head of the peeling sequence and initialize its peeling weight by $\Delta_0 = 0$.

\stitle{Insertion of an edge $(u_i,u_j)$.} Without loss of generality, we assume $i < j$ and denote the weight of $(u_i,u_j)$ by $\Delta = c_{ij}$. \jiaxin{Given an edge insertion $(u_i,u_j)$, we observe that a part of the peeling sequence will not be changed. We formalize the finding as follows.}

\begin{lemma}\label{lemma:seq}
    $\Seq'[1:i-1] = \Seq[1:i-1]$.
\end{lemma}

\jiaxin{Due to space limitations, all the proofs in this section are presented in Appendix~\ref{sec:proof} of \cite{techreport}.}

\stitle{Affected area (\AFF) and pending queue ($T$).} Given updates $\Delta G$ to graph $G$ and an incremental algorithm \TCal, we denote by \AFF{}$=(V_{\mathcal{T}},E_{\mathcal{T}})$ the subgraph inspected by \TCal{} in $G$ that indicates the necessary cost of incrementalization. Moreover, we construct a priority queue $T$ for the vertices pending reordering in ascending order of the peeling weights.

\stitle{Incremental algorithm (\TCal).} \TCal{} initializes an empty vector for the updated peeling sequence $\Seq'$ and append $\Seq[1:i-1]$ to $\Seq'$ due to the Lemma~\ref{lemma:seq}. We iteratively compare 1) the head of $T$, denoted by $u_{\min}$ and 2) the vertex $u_k$ in the peeling sequence $\Seq$, where $k > i$. The corresponding peeling weights are denoted by $\Delta_{\min}$ and $\Delta_k$. We consider the following three cases:

\stitle{Case 1.} If $\Delta_{\min} < \Delta_k $, we pop the $u_{\min}$ from $T$ and insert it to $\Seq'$. Then we update the priorities in $T$ for the neighbors of $u_{\min}$, $N(u_{\min})$.

\stitle{Case 2.} If $\Delta_{\min} \geq \Delta_{k}$ and $\exists u_{T}\in T, (u_{T},u_{k})\in E$ or $(u_{k},u_{T})\in E$, we insert $u_k$ into $T$. The peeling weight is $w_{u_k}(T\cup S_k)$ $ = \Delta_k $ $ + $ $\sum_{(u_T\in T) \bigwedge ((u_T,u_k)\in E)}$ $ c_{Tk} + $ $\sum_{(u_T\in T) \bigwedge ((u_k,u_T)\in E)} c_{kT}$, $k=k+1$.

\stitle{Case 3.} If $\Delta_{\min} \geq \Delta_{k}$ and $\forall u_{T}\in T, (u_{T},u_{k})\not\in E$ and $(u_{k},u_{T})\not\in E$, we insert $u_{k}$ to $\Seq'$, $k=k+1$.

\setlist{nolistsep}
\eat{
\begin{itemize}[leftmargin=*]
    \item \stitle{Case 1.} If $\Delta_{\min} < \Delta_k $, we pop the $u_{\min}$ from $T$ and insert it to $\Seq'$. Then we update the priorities in $T$ for the neighbors of $u_{\min}$, $N(u_{\min})$.
    \item \stitle{Case 2.} If $\Delta_{\min} \geq \Delta_{k}$ and $\exists u_{T}\in T, (u_{T},u_{k})\in E$ or $(u_{k},u_{T})\in E$, we insert $u_k$ into $T$. The peeling weight is $w_{u_k}(T\cup S_k)$ $ = \Delta_k $ $ + $ $\sum_{(u_T\in T) \bigwedge ((u_T,u_k)\in E)}$ $ c_{Tk} + $ $\sum_{(u_T\in T) \bigwedge ((u_k,u_T)\in E)} c_{kT}$, $k=k+1$.
    \item \stitle{Case 3.} If $\Delta_{\min} \geq \Delta_{k}$ and $\forall u_{T}\in T, (u_{T},u_{k})\not\in E$ and $(u_{k},u_{T})\not\in E$, we insert $u_{k}$ to $\Seq'$, $k=k+1$.
\end{itemize}
}

We repeat the above iteration until $T$ is empty.

\begin{figure*}[tb]
\includegraphics[width=0.95\linewidth]{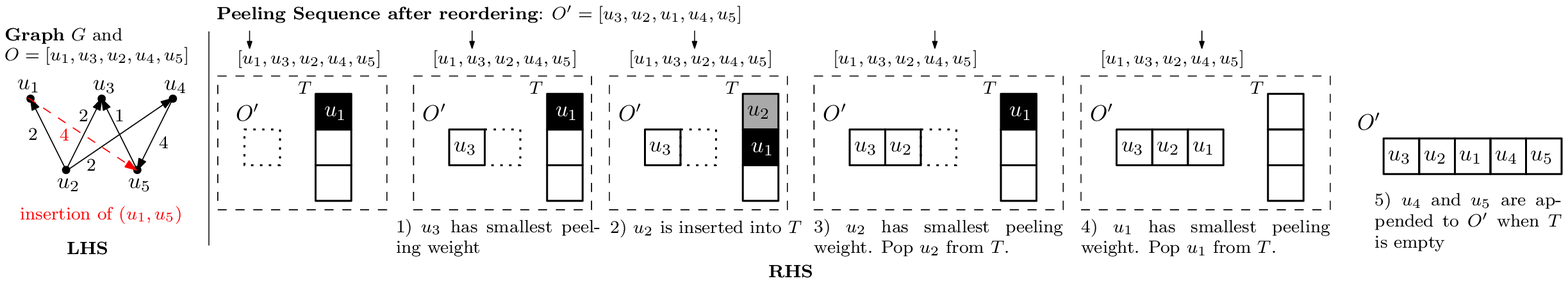}
\vspace{-1.4em}
    \caption{\jiaxin{Peeling sequence reordering  with edge insertion (A running example)} }\label{fig:reorder}
\end{figure*}

\begin{example}
    Consider the graph $G$ in Figure~\ref{fig:peel} and its peeling sequence $\Seq = [u_1, u_3, u_2, u_4, u_5]$. Suppose that a new edge $(u_1,u_5)$ is inserted into $G$ and its weight is $4$ as shown in the LHS of Figure~\ref{fig:reorder}. The reordering procedure is presented in the RHS of Figure~\ref{fig:reorder}. $u_1$ is pushed to the pending queue $T$. Since the peeling weight of the next vertex in $\Seq$, $u_3$, is the smallest, it will be inserted directly into $\Seq'$. Since $u_2\in N(u_1)$, we recover its peeling weight and push it into $T$. Since the peeling weights of $u_2$ and $u_1$ are smaller than those of $u_4$, they will pop out of $T$ and insert into $\Seq'$. Once $T$ is empty, the rest of the vertices, $u_4$ and $u_5$, in $\Seq$ are appended to $\Seq'$ directly. Therefore, the reordered peeling sequence is $\Seq'=[u_3,u_2,u_1,u_4,u_5]$.
\end{example}

\eat{
\begin{lemma}\label{lemma:subset}
    If $S_i \subseteq S_j$ and $u_k\in S_i$, $w_{u_k}(S_j) \geq w_{u_k}(S_i)$.
\end{lemma}
}

\eat{
\begin{proof}
By definition, we have the following.
\begin{equation}
\footnotesize
\begin{split}
    w_{u_k}(S_j) & = a_k + \sum_{(u_j\in S_j) \bigwedge ((u_k,u_j)\in E)} c_{kj} + \sum_{(u_j\in S_j) \bigwedge ((u_j,u_k)\in E)} c_{jk} \\
    & = w_{u_k}(S_i) + \sum_{(u_j\in S_j\setminus S_i) \bigwedge ((u_k,u_j)\in E)} c_{kj} + \sum_{(u_j\in S_j\setminus S_i) \bigwedge ((u_j,u_k)\in E)} c_{jk}
\end{split}
\end{equation}
Since the weights on the edges are nonnegative, $w_{u_k}(S_j) > w_{u_k}(S_i)$.
\end{proof}
}

\noindent\jiaxin{\stitle{Remarks.} If the peeling weight of $u_k$ is greater than that of the head of $T$ (\ie $u_{\min}$), then $u_{\min}$ has the smallest peeling weight among $T\cup S_k$. We formalize this remark as follows.}

\begin{lemma}\label{lemma:peel}
    If $\Delta_k > \Delta_{\min}$, $u_{\min} = \mathop{\arg\min}\limits_{u\in T\cup S_k}w_{u}(T\cup S_k)$.
\end{lemma}
\eat{
\begin{proof}
\jiaxin{Due to space limitations, the proof is presented in Appendix~\ref{sec:proof} of \cite{techreport}.}
\end{proof}
}

\eat{
\begin{proof}
    Consider a vertex $u'\in T\cup S_k$, where $u'\not =u_k$ or $u'\not = u_{\min}$. 1) If $u'\in S_k$, due to Lemma~\ref{lemma:subset}, $w_{u'}(T\cup S_k) > w_{u'}(S_k) > w_{u_k}(S_k) \geq w_{u_k}(T\cup S_k) = \Delta_k >\Delta_{\min}$. 2) If $u'\in T$, $w_{u'}(T\cup S_k) > w_{u_{\min}}(T\cup S_k) = \Delta_{\min}$. Hence, $u'$ is not the vertex that has the smallest peeling weight. Therefore, $u_{\min}$ has the smallest peeling weight.
\end{proof}
}

\stitle{Correctness and accuracy guarantee.} In \textbf{Case 1} of \TCal{}, if $\Delta_k > \Delta_{\min}$, $u_{\min}$ is chosen to insert to $\Seq'$ since it has the smallest peeling weight due to Lemma~\ref{lemma:peel}. In \textbf{Case 3} of \TCal{}, $\Delta_k$ is the smallest peeling weight and $u_k$ is chosen to insert to $\Seq'$. The peeling sequence is identical to that of $G\oplus \Delta G$, since in each iteration the vertex with the smallest peeling weight is chosen. The accuracy of the worst-case is preserved due to Lemma~\ref{lemma:2ppr}.

\stitle{Time complexity.} The complexity of the incremental maintenance is $O(|E_{\mathcal{T}}|+|E_{\mathcal{T}}|\log|V_{\mathcal{T}}|)$. The complexity is bounded by $O(|E|\log |V|)$ and is small in practice.

\subsection{Peeling sequence reordering in batch} \label{sec:batch}

\begin{figure}
    \includegraphics[width=0.9\linewidth]{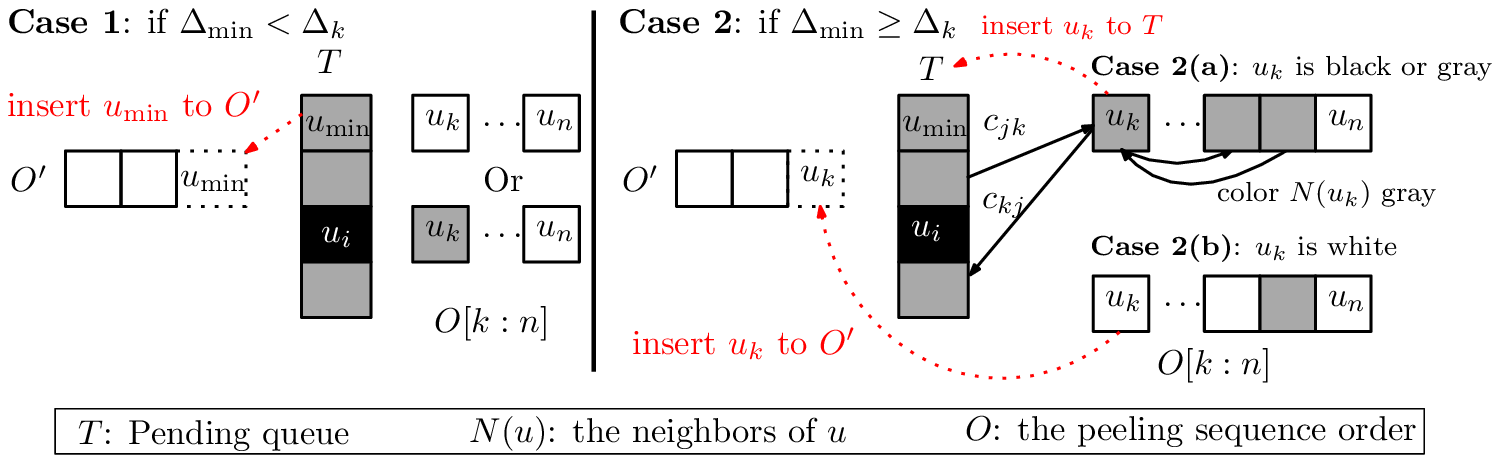}
    \caption{\jiaxin{Peeling sequence reordering in batch}}\label{fig:batch}
\end{figure}

Since the peeling sequence reordering by early edge insertions could be reversed by later ones, some reorderings are stale and duplicate. Suppose that the insertion is a subgraph $\Delta G = (\Delta V, \Delta E)$. A direct way to reorder the peeling sequence is to insert the edges one by one. The complexity is $O(|\Delta E| (|E_{\mathcal{T}}|\log|V_{\mathcal{T}}|))$ which is time consuming. To reduce the amount of stale computation, we propose a peeling sequence reordering algorithm in batch.

\begin{figure}
    \includegraphics[width=0.65\linewidth]{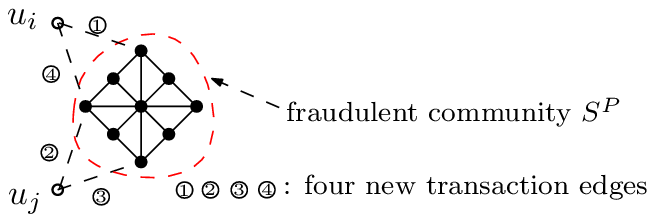}
    \caption{Illustration of stale incremental maintenance}\label{fig:stale}
\end{figure}

\begin{example}\label{eg:stale}
    Consider a fraudulent community, $S^P$, identified by the peeling algorithm in Figure~\ref{fig:stale}. $u_i$ and $u_j$ are two normal users. Suppose that they have the same peeling weight and that $u_i$ is peeled before $u_j$. When a new transaction \textcircled{\raisebox{-0.9pt}{1}} is generated, we should reorder $u_i$ and $u_j$ by exchanging their positions. When \textcircled{\raisebox{-0.9pt}{2}} and \textcircled{\raisebox{-0.9pt}{3}} are inserted, positions of $u_i$ and $u_j$ will be re-exchanged. However, if we reorder the sequence in batch with the last transaction \textcircled{\raisebox{-0.9pt}{4}}, we are not required to change the positions of $u_i$ and $u_j$.
\end{example}

\begin{algorithm}[tb]
    \caption{Peeling sequence reordering in batch}\label{algo:batch}
    \footnotesize
    \SetKwProg{Fn}{Function}{}{}
    \KwIn{Graph $G = (V, E)$, $\Seq$, density metric $g(S)$, $\Delta G = (\Delta V, \Delta E)$}
    \KwOut{Peeling sequence order $\Seq' = Q(G\oplus \Delta G)$ and fraudulent community}

    sort $\Delta V$ in the ascending order of indices in $\Seq$ and color $\Delta V$ black \\
    
    init a priority pending queue $T$ in the ascending order of peeling weights \\ \label{algo:batch:init}

    init an empty vector $\Seq'$ \\

    \For{$u_i = \Seq[i] \in \Delta V$}{
        add $u_i$ into $T$ \\\label{algo:batch:add1}
        color its neighbors $\Seq[j]$ ($j>i$) gray\\\label{algo:batch:add2}
        $k = i + 1$ \\
        \While{$T$ is not empty}{
         \eIf(\jiaxin{\tcp*[h]{\textbf{Case 1}}}){\label{algo:batch:case1b}
         $\Delta_{\min} < \Delta_k$}{ 
            pop $u_{\min}$ from $T$ and insert it to $\Seq'$ \\
                    update the priorities of $N(u_{\min})$ in $T$ \\\label{algo:batch:case1e}
            }{\label{algo:batch:case2ab}
                \eIf(\jiaxin{\tcp*[h]{\textbf{Case 2(a)}}}){$u_k$ is black or gray}{
                   add $u_k$ into $T$ and recover its peeling weight\\
                color its neighbors $N(u_{k})$ gray \\ \label{algo:batch:case2ae}
            }(\jiaxin{\tcp*[h]{\textbf{Case 2(b)}}: $u_k$ is white}){\label{algo:batch:case2bb}
                    insert $u_k$ to $\Seq'$\\
                }
                $k = k + 1$ \\\label{algo:batch:case2be}
            }
        }
        append $\Seq[k:i'-1]$ to $\Seq'$, where $u_{i'}=\Seq[i']$ is the next black vertex \\
    }

    \Return{$\Seq'$ \textnormal{and} $\arg\max_{S_i}g(S_i)$}
\end{algorithm}

\stitle{Peeling weight recovery.} Given a vertex $u_j = \Seq[j]$ and a set of vertex $S_i$ ($i < j$, \ie $S_j\subseteq S_i$), the peeling weight $w_{u_j}(S_i)$ can be calculated by $w_{u_j}(S_i) = \Delta_j + \sum_{(i \leq k < j) \bigwedge ((u_j,u_k)\in E)} c_{jk} + \sum_{(i \leq k < j) \bigwedge ((u_k,u_j)\in E)} c_{kj}$.

\stitle{Vertex sorting.} Intuitively, the increase in peeling weight of $u_i$ does not change the subsequence of $\Seq[1:i-1]$ due to Lemma~\ref{lemma:seq}. We sort the vertices in $\Delta V$ by the indices in the peeling sequence. Then we reorder the vertices in ascending order of the indices in $\Seq$. For simplicity, we color the vertices in $\Delta V$ black, affected vertices (\ie vertices pending reordering) gray and unaffected vertices white.

\stitle{Incremental maintenance in batch \jiaxin{(Algorithm~\ref{algo:batch} and Figure~\ref{fig:batch})}.} We initialize a pending queue $T$ to maintain the vertices pending reordering (Line~\ref{algo:batch:init}). Iteratively, we add the vertex $\Seq[i]\in \Delta V$ to $T$ and color its neighbors $\Seq[j]$ gray (Line~\ref{algo:batch:add1}-\ref{algo:batch:add2}). If $T$ is not empty, we compare the peeling weight $\Delta_k$ of the vertex $u_k = \Seq[k]$ ($k > i$) with the peeling weight $\Delta_{\min}$ of the head of $T$, $u_{\min}$. We consider the following \jiaxin{two cases as shown in Figure~\ref{fig:batch}}. \textbf{Case 1:} If $\Delta_{\min} < \Delta_k$, we pop $u_{\min}$ from $T$, insert it to $\Seq'$ and update the priorities of its neighbors in $T$ \jiaxin{(Line~\ref{algo:batch:case1b}-\ref{algo:batch:case1e})}; \textbf{Case 2(a):} if $\Delta_{\min} \geq \Delta_k$ and $u_k$ is gray or black, we recover its peeling weight in $S_k\cup T$ and insert it to $T$. Then we color the vertices in $N(u_k)$ gray \jiaxin{(Line~\ref{algo:batch:case2ab}-\ref{algo:batch:case2ae})}; otherwise \textbf{Case 2(b):} if $\Delta_{\min} \geq \Delta_k$ and $u_k$ is white, we insert $u_k$ to $\Seq'$ directly \jiaxin{(Line~\ref{algo:batch:case2bb}-\ref{algo:batch:case2be})}. We repeat the above procedure until the pending queue $T$ is empty. Then we append $\Seq[k:i'-1]$ to $\Seq'$, where $u_{i'}$ is the next vertex in $\Delta V$. We insert $u_{i'}$ into $T$ and repeat the reordering until there is no black vertex. \jiaxin{The correctness and accuracy guarantee are similar to those of peeling sequence reordering with edge insertion. Due to space limitations, we present them in Appendix~\ref{sec:correctness} of ~\cite{techreport}.}

\stitle{Complexity.} The time complexity of Algorithm~\ref{algo:batch} is $O(|E_{\mathcal{T}}|+|E_{\mathcal{T}}| $ $\log|V_{\mathcal{T}}|)$ which is bounded by $O(|E|\log|V|)$.

\subsection{Peeling sequence reordering with edge grouping}~\label{sec:urgent}
\stitle{Update steam $\Delta G^{\tau}$.} In a transaction system, the edge updates are coming in a stream manner (\ie a timestamp on each edge) which is denoted by $\Delta G^{\tau}$. Formally, we denote it by $\Delta G^{\tau} = [(e_0,\tau_0),\ldots (e_n, \tau_n)]$ where $\tau_i$ is the timestamp on the edge $e_i=(u_i,v_i)$.

\stitle{Latency of activities $\latency(\Delta G^{\tau})$.} Suppose that $e_i=(u_i,v_i)$ is a labeled fraudulent activity which is generated at $\tau_i$ and is responded/inserted at $\tau_i^r$. The latency of $e_i$ is $\tau_i^r - \tau_i$. Given an update stream $\Delta G^{\tau}$, the latency of fraudulent activities is defined as follows.
\begin{equation}\label{eq:latency}
    \latency(\Delta G^{\tau}) = \sum_{(e_i,\tau_i)\in \Delta G^{\tau}}{\tau_i^r - \tau_i}
\end{equation}

\stitle{Prevention ratio $\Ratio$.} If a fraudster is identified, we ban the following related  transactions to prevent economic loss. We denote the ratio of suspicious transactions prevented to all suspicious transactions by $\Ratio$.

\jiaxin{
\begin{example}
    Consider an update steam in Figure~\ref{fig:effectivepara}. $e_i$ ($i\in [1,6]$) are a set of labeled fraudulent transactions and $\tau_i$ ($i\in [1,6]$) are their timestamps. Regarding the reordering in batch, the new transactions are queueing until the size of the queue is equal to the batch size. The reordering is triggered at $\tau_s$ and finished at $\tau_f$. Therefore, they are inserted at $\tau_i^r = \tau_f$ The queueing time for each edge is $\tau_s - \tau_i$ while the latency is $\tau_f - \tau_i$. Suppose the fraudster is identified at $\tau_f$, the prevention ratio is $\Ratio = \frac{|\{e_i|\tau_i > \tau_f\}|}{|\{e_i\}|}$.
\end{example}
}

\Spade{} aims to reduce $\latency$ and increase $\Ratio$ as much as possible. \jiaxin{In Figure~\ref{fig:effectivepara}, if the reordering is triggered at $\tau_s=\tau_2$ and responded at $\tau_f=\tau_3$, the following fraudulent activities can be prevented.}
 
Intuitively, some transactions are generated by normal users (benign edges), while others are generated by potential fraudsters (urgent edges). \Spade{} groups the benign edges and reorders the peeling sequence in batch. It can both improve the performance of reordering and reduce the latency of the response to potential fraudulent transactions. We define the benign and urgent edges as follows.

\begin{definition}
    Given an edge $e=(u_i, u_j)$ and its weight $c_{ij}$, if $w_{u_i}(S_0) + c_{ij} \geq g(S^P)$ or $w_{u_j}(S_0) + c_{ij} \geq g(S^P)$, $e$ is an urgent edge; otherwise $e$ is a benign edge.
\end{definition}

\noindent\jiaxin{Given a benign edge insertion $(u_i,u_j)$, neither $u_i$ nor $u_j$ belongs to the densest subgraph (Lemma~\ref{lemma:optimalset}). And the insertion cannot produce a denser fraudulent community by peeling algorithms (Lemma~\ref{lemma:benign}).}

\eat{
\begin{lemma}\label{lemma:opt}
    If $\exists u\in S$, such that $w_{u}(S) < g(S^*)$, then $S\not = S^{*}$. 
\end{lemma}
}

\eat{
\begin{proof}
    
\end{proof}
}

\begin{lemma}\label{lemma:optimalset}
Given an edge $e = (u_i,u_j)$, if $e$ is a benign edge, after the insertion of $e$, $u_i\not \in S^{*}$ and $u_j\not \in S^{*}$.
\end{lemma}
\eat{
\begin{proof}
    We prove this lemma in contradiction by assuming that $u_i\in S^{*}$. $w_{u_i}(S^*) \leq w_{u_i}(S_0) + c_{ij} < g(S^P) \leq g(S^*)$. We have $S^*\not = S^*$ due to Lemma~\ref{lemma:opt}. We can conclude that $u_i\not\in S^*$. Similarly, $u_j\not \in S^*$.
\end{proof}
}

\eat{
\begin{lemma}\label{lemma:opt2}
    If $\exists u\in S_i$, $w_{u}(S_i) < g(S_i)$, then $S_i\not = S^{P}$.
\end{lemma}
}

\eat{
\begin{proof}
We prove this in contradiction by assuming that $S_i= S^P$. Suppose that $u_i$ is peeled from $S_i$. Hence, $w_{u_i}(S^P) \leq w_{u}(S^P)$ due to the peeling definition. The proof can be obtained as follows:
\begin{equation}
\footnotesize
    \begin{split}
        g(S^P \setminus \{u_i\}) & = \frac{f(S^{P}) - w_{u_i}(S^{P})}{|S^{P}|-1} > \frac{f(S^P) - w_{u}(S^P)}{|S^P|-1} > g(S^P)   
    \end{split}
\end{equation}
This contradicts the fact that $S^P$ has the highest density. We can conclude that $S_i \not = S^{P}$.
\end{proof}
}

We denote the vertex subset returned after reordering by $S^{P'}$.

\begin{lemma}\label{lemma:benign}
Given a benign edge $e = (u_i,u_j)$ insertion, at least one of the following two conditions is established: 1) $u_i\not \in S^{P'}$ and $u_j\not \in S^{P'}$; and 2) $g(S^{P'}) < g(S^P)$.
\end{lemma}

Therefore, we postpone the incremental maintenance of the peeling sequence for benign edges which provide two benefits. First, we can perform a batch update that avoids stale computation. Second, an urgent edge insertion, which is caused by a potential fraudster, triggers incremental maintenance immediately. These fraudsters are identified and reported to the moderators in real time.

\begin{figure}[tb]
    \includegraphics[width=0.75\linewidth]{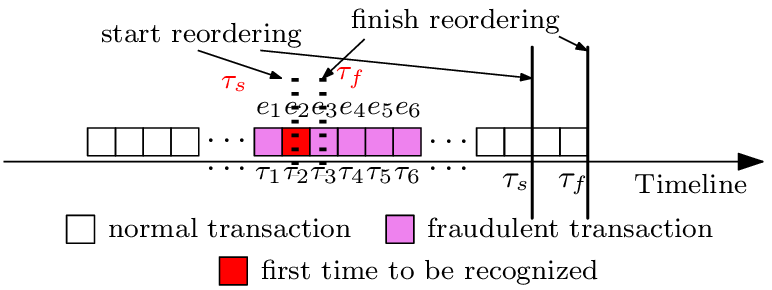}
    \caption{\small \jiaxin{Metrics for a set of fraudulent transactions made by a fraudster (latency: $\tau_f-\tau_i$, queueing time: $\tau_s-\tau_i$, prevention ratio: $\Ratio = \frac{|\{e_i|\tau_i > \tau_f\}|}{|\{e_i\}|}$)}}\label{fig:effectivepara}
\end{figure}

\stitle{Edge grouping.} We next present the paradigm of peeling sequence reordering by edge grouping. We first initialize an empty buffer $\Delta G$ for the updates (Line~\ref{algo:urgent:init}). When an edge $e_i$ enters, we insert it into $\Delta G$. If $e_i$ is an urgent edge, we incrementally maintain the peeling sequence by Algorithm~\ref{algo:batch} and clear the buffer (Line~\ref{algo:urgent:begin}-\ref{algo:urgent:end}).

\begin{algorithm}[tb]
    \caption{Paradigm of edge grouping}\label{algo:urgent}
    \footnotesize
    \SetKwProg{Fn}{Function}{}{}
    \KwIn{A graph $G = (V, E)$, $\Seq$, a density metric $g(S)$, $\Delta G^T$}
    \KwOut{Peeling sequence order $\Seq' = Q(G\oplus \Delta G^T)$ and fraudulent community}

    init an empty buffer $\Delta G$ for updates  \\ \label{algo:urgent:init}
    \For{$i=1,\ldots, m$}{
        $\Delta G.\mathsf{add}$($e_i$) \\
         \If{$e_i$ is an urgent edge}{\label{algo:urgent:begin}
            $\Seq' = Q(G\oplus \Delta G)$ by Algorithm~\ref{algo:batch} \\
            clear $\Delta G$ \label{algo:urgent:end}
         }    
    }
    \Return{$\Seq'$ \textnormal{and} $\arg\max_{S_i}g(S_i)$}
\end{algorithm}

\eat{
\stitle{\ref{sec:Spade}.2 Approximation ratio perserving}
Another approach is that we maintain the same approximation ratio of $\Fraudar$ rather than reorder the sequence $\Fraudar$.
\idea{\cite{tsourakakis2015k} -- peeling algorithms in batch?}
}

\section{Experimental Evaluation}\label{sec:exp}

\eat{Evaluations are classified into two groups: the overall improvement in performance of \Spade{} (Section~\ref{sec:efficiency}) and the effectiveness of \Spade{} in preventing fraudulent transactions (Section~\ref{sec:effectiveness}).}

Our experiments are run on a machine that has an X5650 CPU, $16$ GB RAM. The implementation is made memory-resident and implemented in C++. All codes are compiled by GCC-9.3.0 with -$O3$.

\begin{table}[tb]
\centering
\begin{scriptsize}
\caption{Statistics of real-world datasets}\label{table:Statistics}
\begin{tabular}{|c|c|c|c|c|c|c|}
  \hline
  {\bf Datasets} & {\bf $|V|$} & {\bf $|E|$}  & avg. degree   & Increments & Type \\
  \hline
    Grab1 & 3.991M & 10M & 5.011 & 1M & Transaction\\
  \hline
  Grab2 & 4.805M & 15M & 6.243 & 1.5M & Transaction\\
  \hline
  Grab3 & 5.433M & 20M & 7.366 & 2M & Transaction\\
  \hline
  Grab4 & 6.023M & 25M & 8.302  & 2.5M & Transaction\\
  \hline
  Amazon~\cite{mcauley2013hidden} & 28K & 28K & 2    & 2.8K & Review \\
  \hline
    Wiki-vote~\cite{leskovec2010signed} & 16K & 103K & 12.88  & 10.3K & Vote\\
  \hline
  Epinion~\cite{leskovec2010signed} & 264K & 841K & 6.37  & 84.1K & Who-trust-whom\\
  \hline
\end{tabular}
\end{scriptsize}
\end{table}

\begin{table*}[tb]
\begin{footnotesize}
\setlength{\tabcolsep}{0.4em}
\begin{tabular}{|c|ccc||ccc|ccc|ccc|ccc|ccc|}
\hline 
                                  & \multicolumn{3}{c||}{Peeling algorithms (seconds)}                                                                  & \multicolumn{3}{c|}{ $|\Delta E| = 1$ ($us$)}                                                           & \multicolumn{3}{c|}{$|\Delta E| = 10$ ($us$)}                                                          & \multicolumn{3}{c|}{$|\Delta E| = 100$ ($us$)}                                                         & \multicolumn{3}{c|}{$|\Delta E| = 1$K ($us$)}                                                        & \multicolumn{3}{c|}{$|\Delta E| = 100$K ($us$)}                                                       \\ \hline \hline
\textbf{Datasets} & \multicolumn{1}{c|}{$\DENG$} & \multicolumn{1}{c|}{$\DENGW$} & $\Fraudar$ & \multicolumn{1}{c|}{\IncDENG} & \multicolumn{1}{c|}{\IncDENGW} & \IncFraudar & \multicolumn{1}{c|}{\IncDENG} & \multicolumn{1}{c|}{\IncDENGW} & \IncFraudar & \multicolumn{1}{c|}{\IncDENG} & \multicolumn{1}{c|}{\IncDENGW} & \IncFraudar & \multicolumn{1}{c|}{\IncDENG} & \multicolumn{1}{c|}{\IncDENGW} & \IncFraudar & \multicolumn{1}{c|}{\IncDENG} & \multicolumn{1}{c|}{\IncDENGW} & \IncFraudar \\ \hline
Grab1                             & \multicolumn{1}{c|}{$12$}               & \multicolumn{1}{c|}{$14$}               & $12$               & \multicolumn{1}{c|}{6517}               & \multicolumn{1}{c|}{17469}              & 6                  & \multicolumn{1}{c|}{3117}               & \multicolumn{1}{c|}{11613}              & 6                  & \multicolumn{1}{c|}{519}                & \multicolumn{1}{c|}{1983}               & 6                  & \multicolumn{1}{c|}{108}                & \multicolumn{1}{c|}{281}                & 6                  & \multicolumn{1}{c|}{5}                 & \multicolumn{1}{c|}{10}                 & 1                      \\ \hline
Grab2                             & \multicolumn{1}{c|}{17}                 & \multicolumn{1}{c|}{20}                 & 16                 & \multicolumn{1}{c|}{6604}               & \multicolumn{1}{c|}{18413}              & 8                  & \multicolumn{1}{c|}{3484}               & \multicolumn{1}{c|}{11280}              & 8                  & \multicolumn{1}{c|}{634}                & \multicolumn{1}{c|}{1782}               & 8                  & \multicolumn{1}{c|}{138}                & \multicolumn{1}{c|}{249}                & 8                  & \multicolumn{1}{c|}{7}                 & \multicolumn{1}{c|}{8}                 & 2                                  \\ \hline
Grab3                             & \multicolumn{1}{c|}{23}                 & \multicolumn{1}{c|}{27}                 & 22                 & \multicolumn{1}{c|}{6716}               & \multicolumn{1}{c|}{18862}              & 11                 & \multicolumn{1}{c|}{3864}               & \multicolumn{1}{c|}{10892}              & 11                 & \multicolumn{1}{c|}{750}                & \multicolumn{1}{c|}{1560}               & 10                 & \multicolumn{1}{c|}{186}                & \multicolumn{1}{c|}{211}                & 10                 & \multicolumn{1}{c|}{8}                 & \multicolumn{1}{c|}{7}                 & 2                                  \\ \hline
Grab4                             & \multicolumn{1}{c|}{27}                 & \multicolumn{1}{c|}{28}                 & 28                 & \multicolumn{1}{c|}{6562}               & \multicolumn{1}{c|}{17469}              & 14                 & \multicolumn{1}{c|}{4108}               & \multicolumn{1}{c|}{11661}              & 12                 & \multicolumn{1}{c|}{878}                & \multicolumn{1}{c|}{1970}               & 13                 & \multicolumn{1}{c|}{206}                & \multicolumn{1}{c|}{267}                & 12                 & \multicolumn{1}{c|}{10}                 & \multicolumn{1}{c|}{9}                 & 3                             \\ \hline
Amazon                            & \multicolumn{1}{c|}{0.49}               & \multicolumn{1}{c|}{0.53}               & 0.43               & \multicolumn{1}{c|}{350}                & \multicolumn{1}{c|}{342}                & 1                  & \multicolumn{1}{c|}{186}                & \multicolumn{1}{c|}{191}                & -                  & \multicolumn{1}{c|}{29}                 & \multicolumn{1}{c|}{30}                 & -                  & \multicolumn{1}{c|}{7}                  & \multicolumn{1}{c|}{6}                  & -                  & \multicolumn{1}{c|}{-}                  & \multicolumn{1}{c|}{-}                  & -                             \\ \hline
Wiki-Vote                         & \multicolumn{1}{c|}{0.022}              & \multicolumn{1}{c|}{0.021}              & 0.017              & \multicolumn{1}{c|}{184}                & \multicolumn{1}{c|}{149}                & 2                  & \multicolumn{1}{c|}{98}                 & \multicolumn{1}{c|}{84}                 & 1                  & \multicolumn{1}{c|}{29}                 & \multicolumn{1}{c|}{28}                 & 1                  & \multicolumn{1}{c|}{5}                  & \multicolumn{1}{c|}{5}                  & -                  & \multicolumn{1}{c|}{-}                  & \multicolumn{1}{c|}{-}                  & -                             \\ \hline
Epinion                           & \multicolumn{1}{c|}{0.25}               & \multicolumn{1}{c|}{0.26}               & 0.23               & \multicolumn{1}{c|}{170}                & \multicolumn{1}{c|}{151}                & 5                  & \multicolumn{1}{c|}{83}                 & \multicolumn{1}{c|}{80}                 & 3                  & \multicolumn{1}{c|}{32}                 & \multicolumn{1}{c|}{30}                 & 2                  & \multicolumn{1}{c|}{10}                 & \multicolumn{1}{c|}{10}                 & 2                  & \multicolumn{1}{c|}{1}                  & \multicolumn{1}{c|}{1}                  & -                                  \\ \hline
\end{tabular}
\end{footnotesize}
\caption{\jiaxin{Time taken for incremental maintenance with \Spade{} by varying batch sizes (avg. time for one edge, - means $ < 1us$)}}\label{tab:batchsize}
\end{table*}

\begin{table*}[tb]
\begin{footnotesize}
\begin{tabular}{|c|cccccc||cccccc|cccccc|}
\hline
                        & \multicolumn{6}{c||}{Peeling algorithms (seconds)}                                                                                                                                                 & \multicolumn{6}{c|}{$|\Delta E| = 1$K ($us$)}                                                                                                                    & \multicolumn{6}{c|}{Edge grouping ($us$)}                                                                                                                        \\ \hline \hline
\multirow{2}{*}{\textbf{Datast}} & \multicolumn{2}{c|}{$\DENG$}                                        & \multicolumn{2}{c|}{$\DENGW$}                                        & \multicolumn{2}{c||}{$\Fraudar$}                   & \multicolumn{2}{c|}{\IncDENG}                                        & \multicolumn{2}{c|}{\IncDENGW}                                        & \multicolumn{2}{c|}{\IncFraudar}                   & \multicolumn{2}{c|}{\IncDENGU}                                        & \multicolumn{2}{c|}{\IncDENGWU}                                        & \multicolumn{2}{c|}{\IncFraudarU}                   \\ \cline{2-19} 
                        & \multicolumn{1}{c|}{$\Elapsed$} & \multicolumn{1}{c|}{$\latency$} & \multicolumn{1}{c|}{$\Elapsed$} & \multicolumn{1}{c|}{$\latency$} & \multicolumn{1}{c|}{$\Elapsed$} & $\latency$ & \multicolumn{1}{c|}{$\Elapsed$} & \multicolumn{1}{c|}{$\latency$} & \multicolumn{1}{c|}{$\Elapsed$} & \multicolumn{1}{c|}{$\latency$} & \multicolumn{1}{c|}{$\Elapsed$} & $\latency$ & \multicolumn{1}{c|}{$\Elapsed$} & \multicolumn{1}{c|}{$\latency$} & \multicolumn{1}{c|}{$\Elapsed$} & \multicolumn{1}{c|}{$\latency$} & \multicolumn{1}{c|}{$\Elapsed$} & $\latency$ \\ \hline
Grab1                   & \multicolumn{1}{c|}{12}      & \multicolumn{1}{c|}{1}       & \multicolumn{1}{c|}{14}      & \multicolumn{1}{c|}{1}       & \multicolumn{1}{c|}{12}      & 1       & \multicolumn{1}{c|}{108}     & \multicolumn{1}{c|}{2.93}       & \multicolumn{1}{c|}{281}       & \multicolumn{1}{c|}{2.51}       & \multicolumn{1}{c|}{6}     & 2.93       & \multicolumn{1}{c|}{24}     & \multicolumn{1}{c|}{0.024}       & \multicolumn{1}{c|}{29}      & \multicolumn{1}{c|}{0.029}       & \multicolumn{1}{c|}{\textbf{5}}     & 0.0042       \\ \hline
Grab2                   & \multicolumn{1}{c|}{17}      & \multicolumn{1}{c|}{1}       & \multicolumn{1}{c|}{20}       & \multicolumn{1}{c|}{1}       & \multicolumn{1}{c|}{16}      & 1       & \multicolumn{1}{c|}{138}     & \multicolumn{1}{c|}{1.37}       & \multicolumn{1}{c|}{249}       & \multicolumn{1}{c|}{1.21}       & \multicolumn{1}{c|}{8}     & 1.43       & \multicolumn{1}{c|}{28}     & \multicolumn{1}{c|}{0.028}       & \multicolumn{1}{c|}{32}      & \multicolumn{1}{c|}{0.032}       & \multicolumn{1}{c|}{\textbf{7}}     &    0.0050    \\ \hline
Grab3                   & \multicolumn{1}{c|}{23}      & \multicolumn{1}{c|}{1}       & \multicolumn{1}{c|}{27}      & \multicolumn{1}{c|}{1}       & \multicolumn{1}{c|}{22}      & 1       & \multicolumn{1}{c|}{186}      & \multicolumn{1}{c|}{0.98}       & \multicolumn{1}{c|}{211}       & \multicolumn{1}{c|}{0.87}       & \multicolumn{1}{c|}{10}      & 1.03       & \multicolumn{1}{c|}{28}     & \multicolumn{1}{c|}{0.028}       & \multicolumn{1}{c|}{29}      & \multicolumn{1}{c|}{0.019}       & \multicolumn{1}{c|}{\textbf{8}}     & 0.0066       \\ \hline
Grab4                   & \multicolumn{1}{c|}{27}      & \multicolumn{1}{c|}{1}       & \multicolumn{1}{c|}{28}      & \multicolumn{1}{c|}{1}       & \multicolumn{1}{c|}{28}      & 1       & \multicolumn{1}{c|}{206}      & \multicolumn{1}{c|}{0.76}       & \multicolumn{1}{c|}{211}       & \multicolumn{1}{c|}{0.74}       & \multicolumn{1}{c|}{10}      & 0.76       & \multicolumn{1}{c|}{29}     & \multicolumn{1}{c|}{0.029}       & \multicolumn{1}{c|}{33}      & \multicolumn{1}{c|}{0.024}       & \multicolumn{1}{c|}{\textbf{10}}     & 0.0073       \\ \hline
\end{tabular}
\end{footnotesize}
\caption{Elapsed time ($\Elapsed$) and latency ($\latency$) of static algorithms, incremental algorithms and edge grouping ($\Elapsed$: The average elapsed time for one edge; $\latency$ is defined by Equation~\ref{eq:latency}.  $\latency$ of \IncDENG{} (resp. \IncDENGW{} and \IncFraudar) is normalized to $\latency$ of $\DENG$ (resp. $\DENGW$ and $\Fraudar$))}\label{tab:urgent}
\end{table*}

\stitle{Datasets.} We conduct the experiments on seven datasets (Table~\ref{table:Statistics}). Four industrial datasets are from $\Grab$ (Grab1-Grab4). \jiaxin{Given a set of transactions, each transaction is represented as an edge. We replay the edges in the \jiaxin{increasing} order of their timestamp. If a user $u_i$ purchases from a store $u_j$, we add an edge $(u_i,u_j)$ to $E$.} Specifically, we construct the graph $G$ as initialization \jiaxin{($V$ and $90\%$ of $E$ as the initial graph), and the remaining $10\%$ of $E$ as increments for testing.} The increments are decomposed into a set of graph updates $\Delta G$ in the increasing order of their timestamp with different batch sizes $|\Delta E|$. We also use three popular open datasets including Amazon~\cite{mcauley2013hidden}, Wiki-vote~\cite{leskovec2010signed} and Epinion~\cite{leskovec2010signed}. Since there are no timestamps on these three datasets, we randomly select $10\%$ edges from $E$ as increments for evaluation.

\stitle{Competitors.} We choose three common peeling algorithms ($\DENG$, $\DENGW$ and $\Fraudar$) as a baseline. \jiaxin{Given an edge insertion, these algorithms identify the fraudulent community on the entire graph from scratch.} We demonstrate the performance improvement of our proposal (\IncDENG{}, \IncDENGW{} and \IncFraudar{}) implemented in \Spade{}. We denote batch updates by \IncDENG{}-$x$, \IncDENGW{}-$x$ and \IncFraudar{}-$x$, where $x =|\Delta E|$ is the batch size. We also denote the reordering of the peeling sequence with edge grouping by \IncDENGU{}, \IncDENGWU{} and \IncFraudarU{}.

\subsection{Efficiency of \Spade{}}\label{sec:efficiency}

\begin{figure}[tb]
    \subcaptionbox{\jiaxin{Prevention ratio vs. latency\label{fig:prevent}}}{\includegraphics[width=0.41\linewidth]{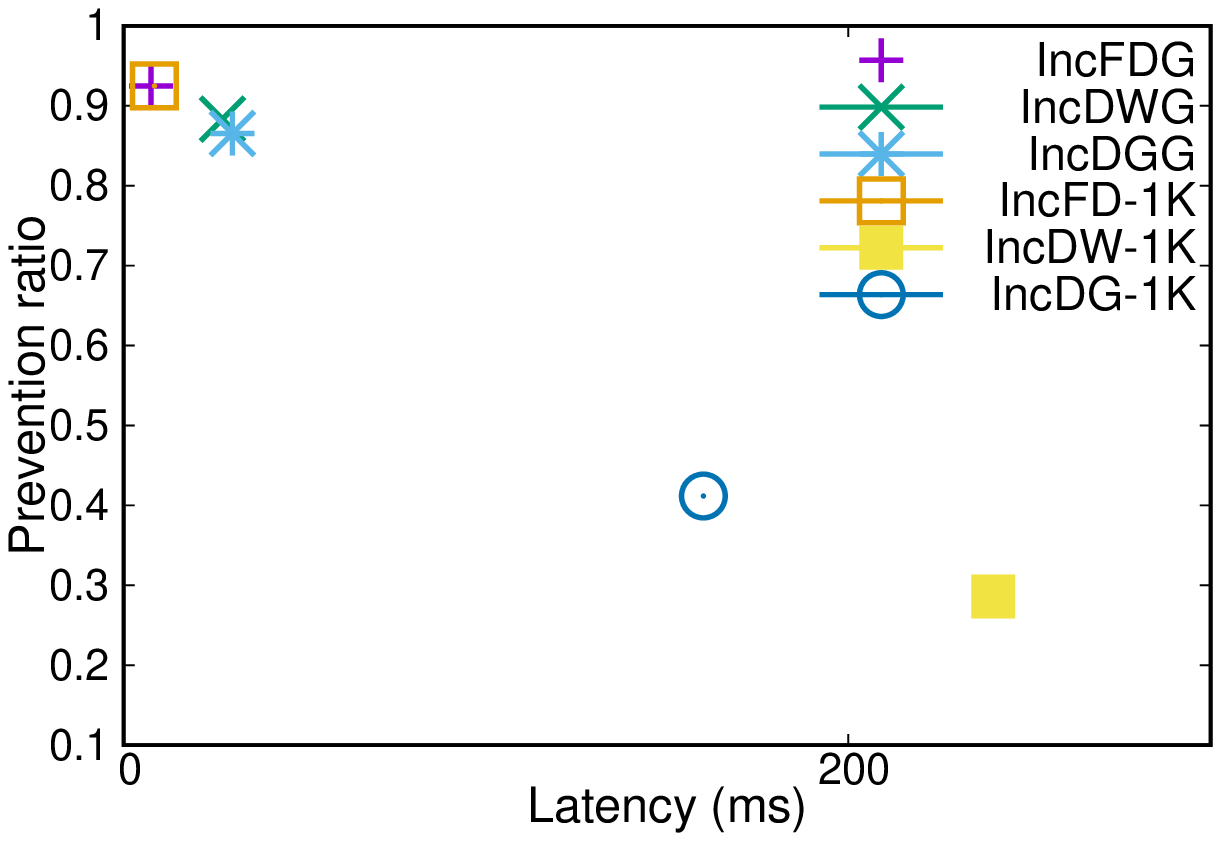}}
    \subcaptionbox{Graph degree distribution\label{fig:distribution}}{\includegraphics[width=0.41\linewidth]{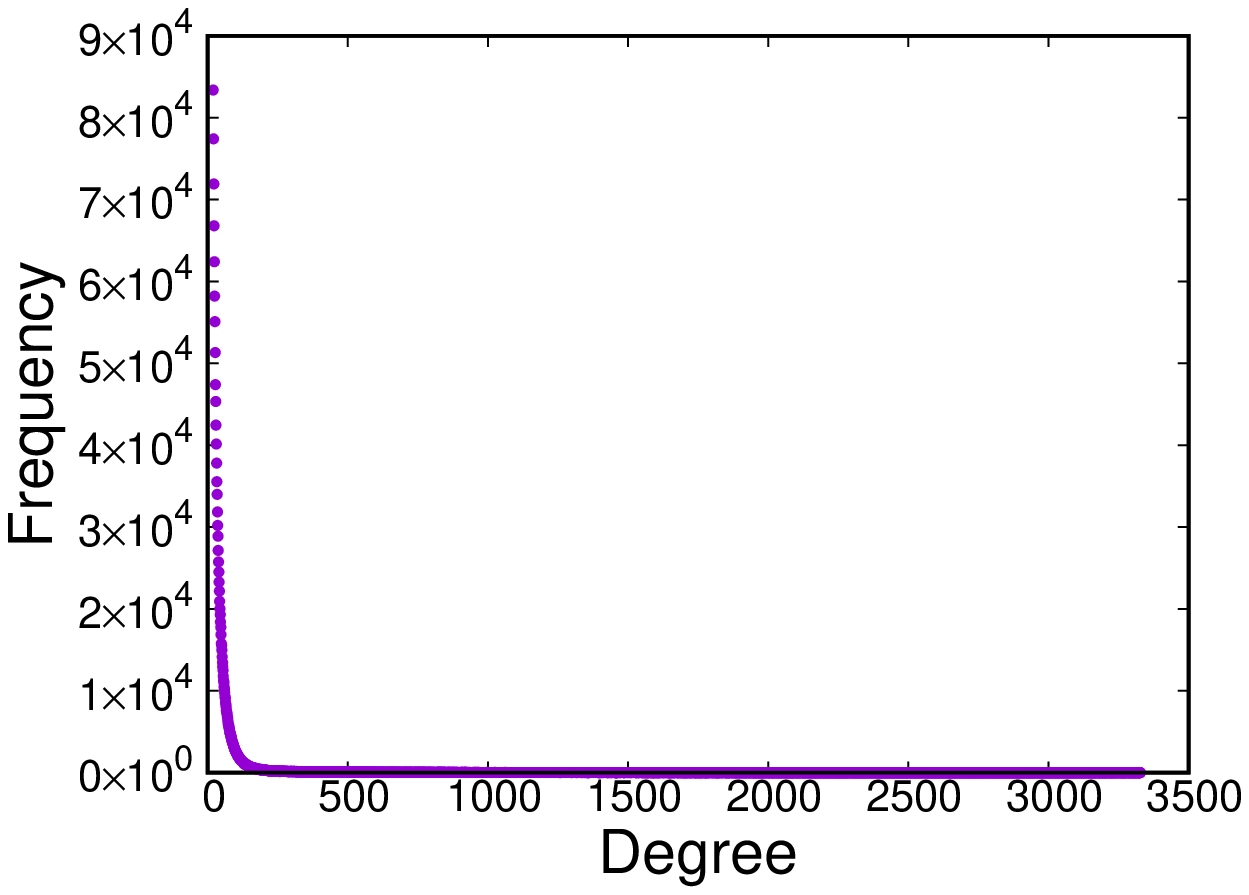}}
    \caption{Graph characteristic}\label{fig:characteristic}
\end{figure}

\begin{figure}[tb]
    \includegraphics[width=0.5\linewidth]{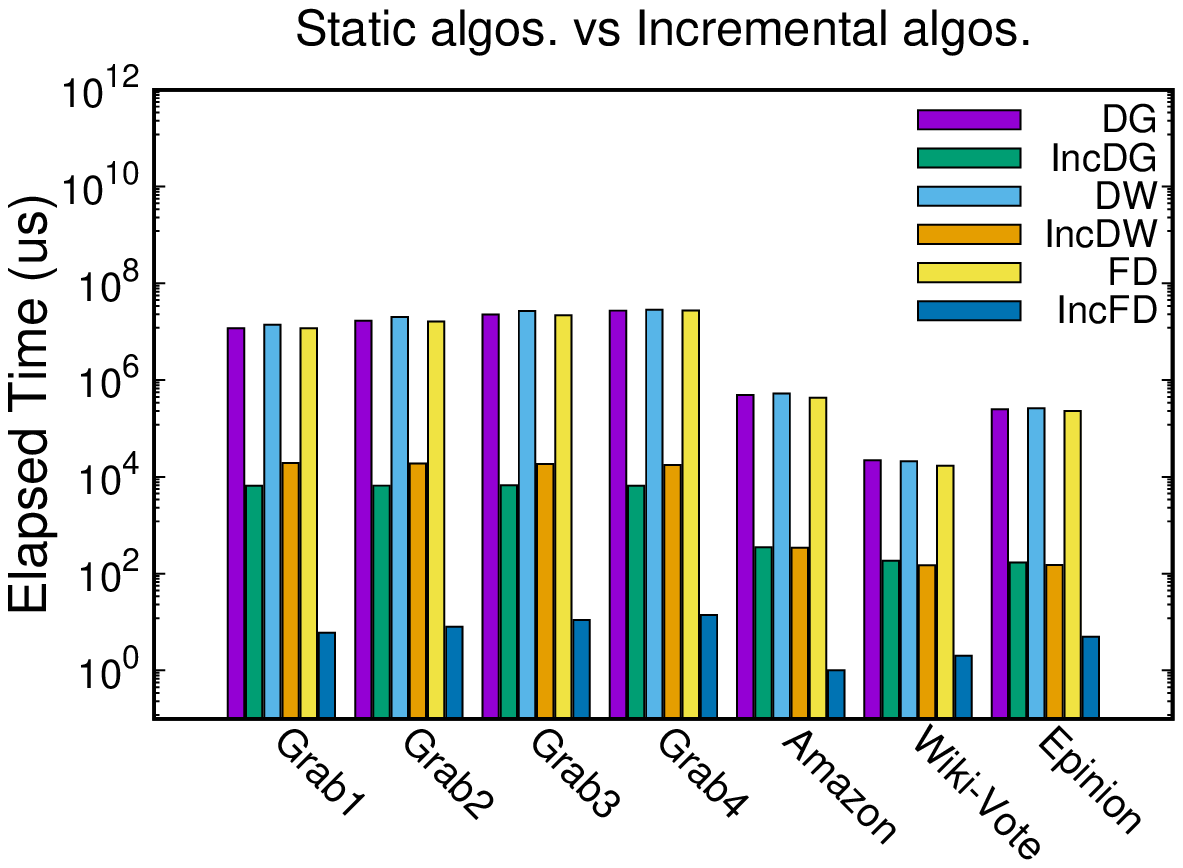}
    \caption{Efficiency comparison between peeling algorithms and corresponding incremental versions on \Spade{} ($|\Delta E| = 1$)}\label{fig:runtime}
\end{figure}

\stitle{Improvement of incremental peeling algorithms.} We first investigate the efficiency of \Spade{} by comparing the performance between incremental peeling algorithms and peeling algorithms. In Figure~\ref{fig:runtime}, our experiments show that \IncDENG{} (resp. \IncDENGW{} and \IncFraudar) is up to $4.17\times 10^3$ (resp. $1.63\times 10^3$ and $1.96\times 10^6$) times faster than $\DENG{}$ (resp. $\DENGW{}$ and $\Fraudar$) with an edge insertion. The reason \jiaxin{for such a significant speedup} is that only a small part of the peeling sequence is affected for most edge insertions. This is also consistent with the time complexity comparison of those algorithms. In fact, our algorithm on average processes only $3.5\times 10^{-4}$, $7.2\times 10^{-4}$ and $2.5\times 10^{-7}$ of edges compared with DG, DW and FD (on the entire graph), respectively. \Spade{} identifies and maintains the affected peeling subsequence rather than recomputes the peeling sequence from scratch. \jiaxin{Thus, \Spade{} significantly outperforms existing algorithms.}

\stitle{Impact of batch sizes $|\Delta E|$.} We evaluate the efficiency of batch updates by varying batch sizes $|\Delta E|$ from $1$ to $100$K. As shown in Table~\ref{tab:batchsize}, \IncDENG-$100$K (resp. \IncDENGW{}-$100$K and \IncFraudar-$100$K) is up to $1211$ (resp. $3448$ and $4.47$) times faster than \IncDENG{} (resp. \IncDENGW{} and \IncFraudar{}). When the batch size increases, the average elapsed time for an edge insertion keeps decreasing. \jiaxin{As indicated in Example~\ref{eg:stale}, the reordering of the peeling sequence by early edge insertions could be reversed by later ones. Reordering the peeling sequence in batch avoids such stale incremental maintenance by reducing the reversal. \tr{With batch updates, \IncDENG-$100$K (resp. \IncDENGW{}-$100$K and \IncFraudar-$100$K) is up to $2.86\times 10^6$ (resp. $3.21\times 10^6$ and $8.8\times 10^6$) times faster than $\DENG$ (resp. $\DENGW$ and $\Fraudar$).}}

\stitle{Impact of edge grouping.} As shown in Table~\ref{tab:urgent}, \IncDENGU{} (resp. \IncDENGWU{} and \IncFraudarU{}) is up to $7.1$ (resp. $9.7$ and $1.25$) times faster than \IncDENG{}-$1$K (resp. \IncDENGW{}-$1$K and \IncFraudar{}-$1$K) since the edge grouping technique generally accumulates more than $1$K edges. Another evidence is that the graph follows the power law, as shown in Figure~\ref{fig:distribution}. Most edge insertions are benign and are processed in batch.

\stitle{Scalability.} We next evaluate the scalability of \Spade{} on $\Grab$' s datasets \jiaxin{(Grab1-Grab4)} of different sizes which is controlled by the number of edges $|E|$. We vary $|E|$ from $10$M to $25$M \jiaxin{as shown in Table~\ref{table:Statistics}} and report the results in Table~\ref{tab:batchsize}. All peeling algorithms scale reasonably well with the increase of $|E|$. With $|E|$ increasing by $2.5$ times, the running time of \Spade{} increases by up to $2$ (resp. $2$ and $3$) times for $\DENG{}$ (resp. $\DENGW{}$ and $\Fraudar$).

We also compare the efficiency of $\DENG$, $\DENGW$ and $\Fraudar$. As shown in Columns $2\sim 4$ of Table~\ref{tab:batchsize}, the peeling algorithms have a similar performance. However, \IncFraudar{} is much faster than \IncDENG{} and \IncDENGW{} since the affected peeling subsequence is smaller due to the suspiciousness function of $\Fraudar$~\cite{hooi2016fraudar}.

\begin{figure}[tb]
\begin{minipage}[t]{.14\textwidth}
    \subcaptionbox{\IncDENG}{\includegraphics[width=1.2\linewidth]{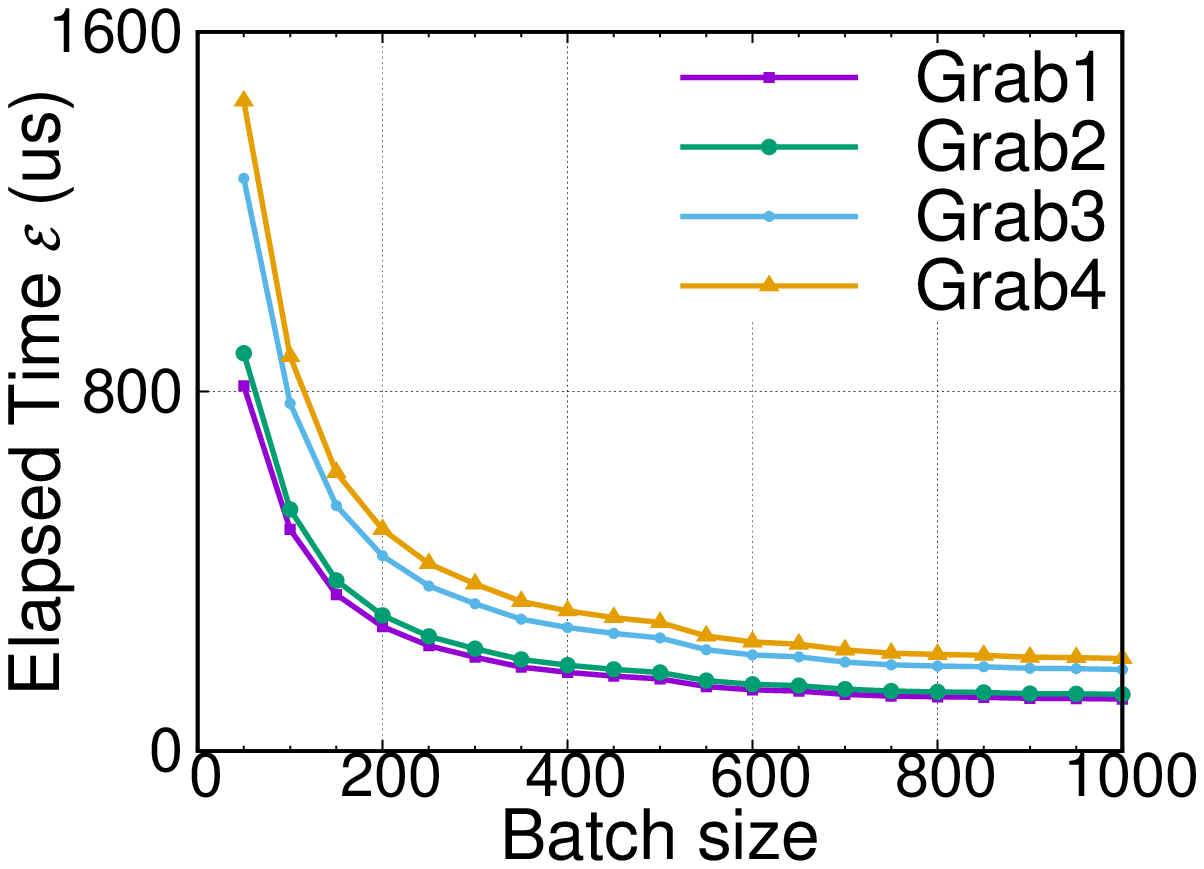}}
\end{minipage}
\hfill\hfill
\begin{minipage}[t]{.14\textwidth}
    \subcaptionbox{\IncDENGW{}}{\includegraphics[width=1.2\linewidth]{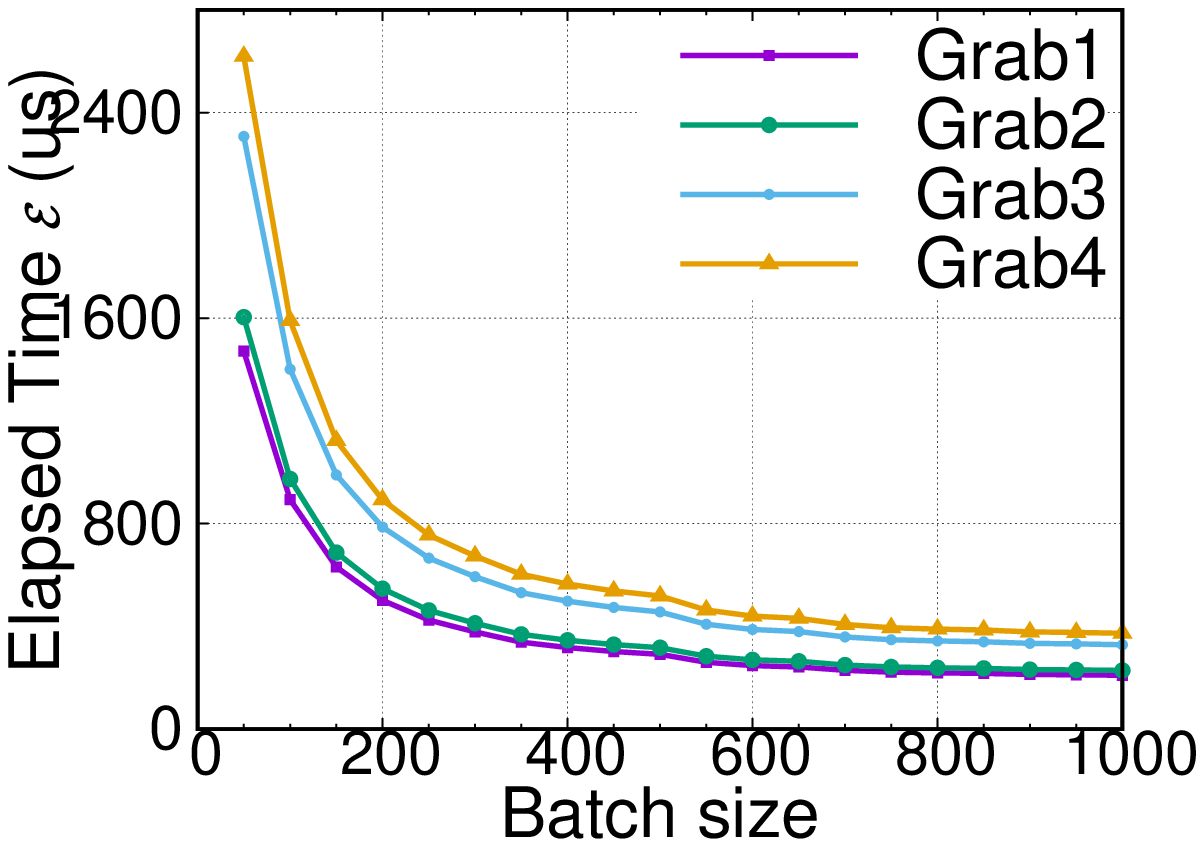}}
\end{minipage}
\hfill\hfill
\begin{minipage}[t]{.14\textwidth}
    \subcaptionbox{\IncFraudar}{\includegraphics[width=1.2\linewidth]{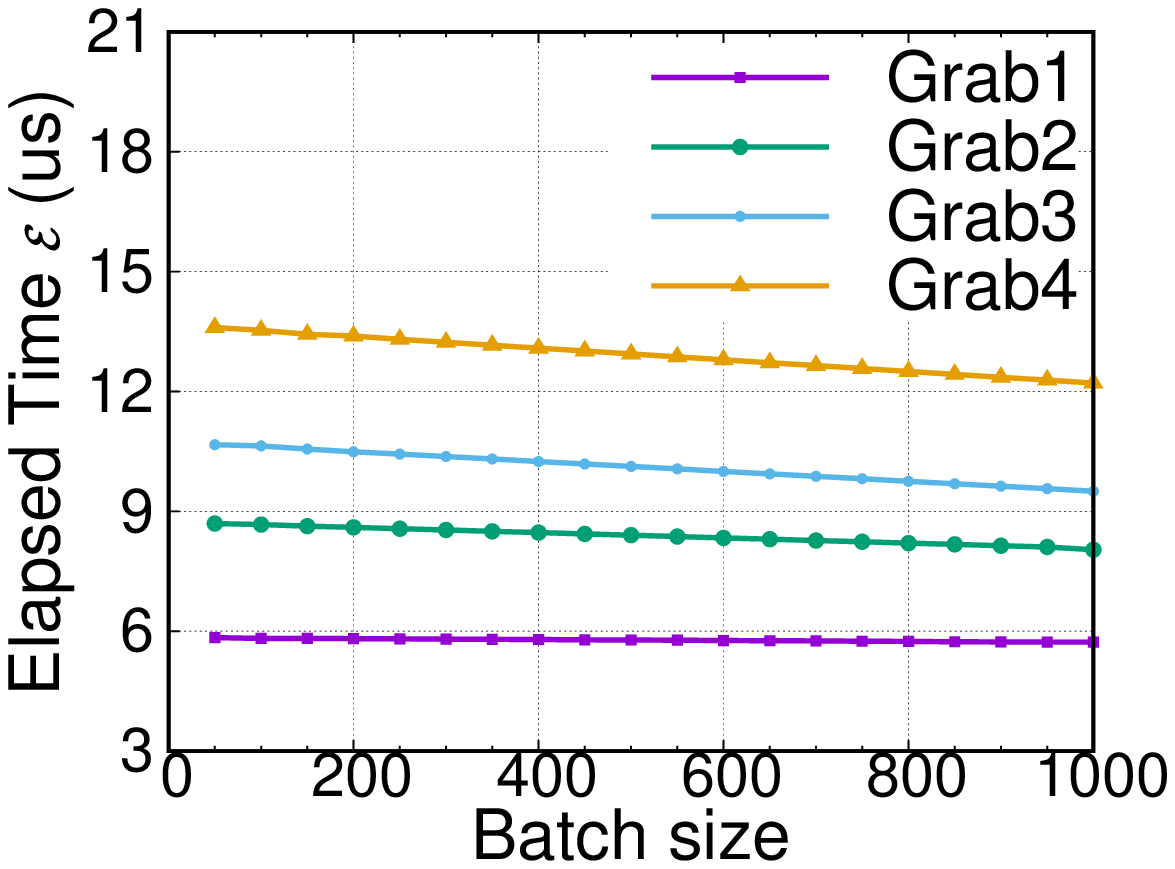}}
\end{minipage}

\begin{minipage}[t]{.14\textwidth}
    \subcaptionbox{\IncDENG}{\includegraphics[width=1.2\linewidth]{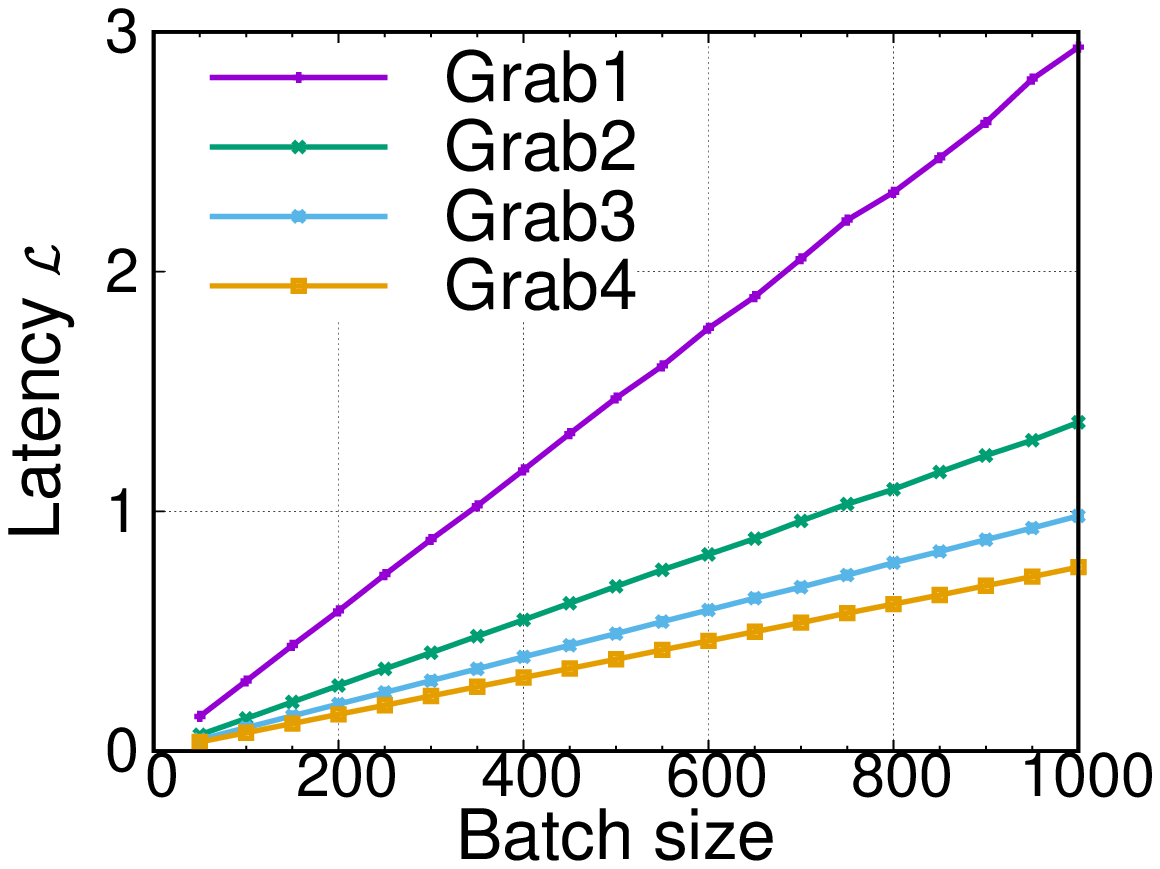}}
\end{minipage}
\hfill\hfill
\begin{minipage}[t]{.14\textwidth}
    \subcaptionbox{\IncDENGW{}}{\includegraphics[width=1.2\linewidth]{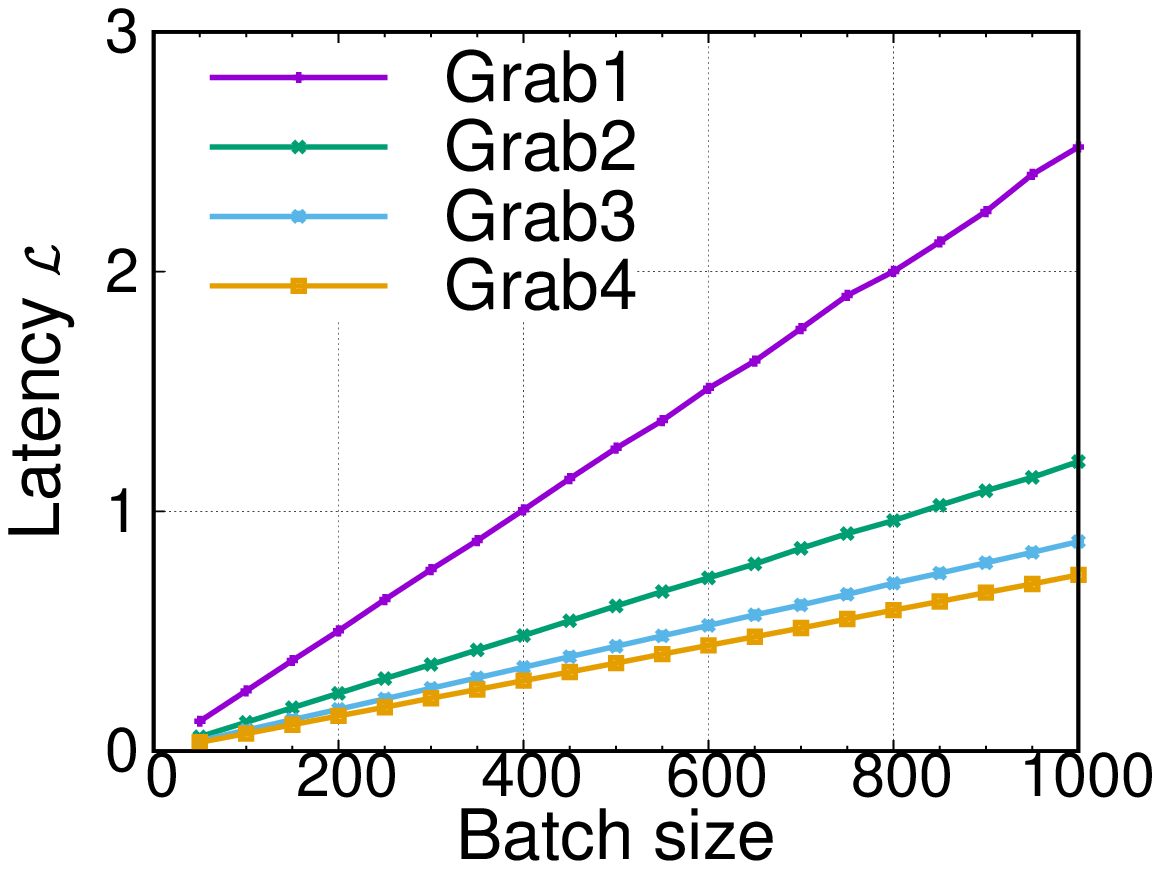}}
\end{minipage}
\hfill\hfill
\begin{minipage}[t]{.14\textwidth}
    \subcaptionbox{\IncFraudar}{\includegraphics[width=1.2\linewidth]{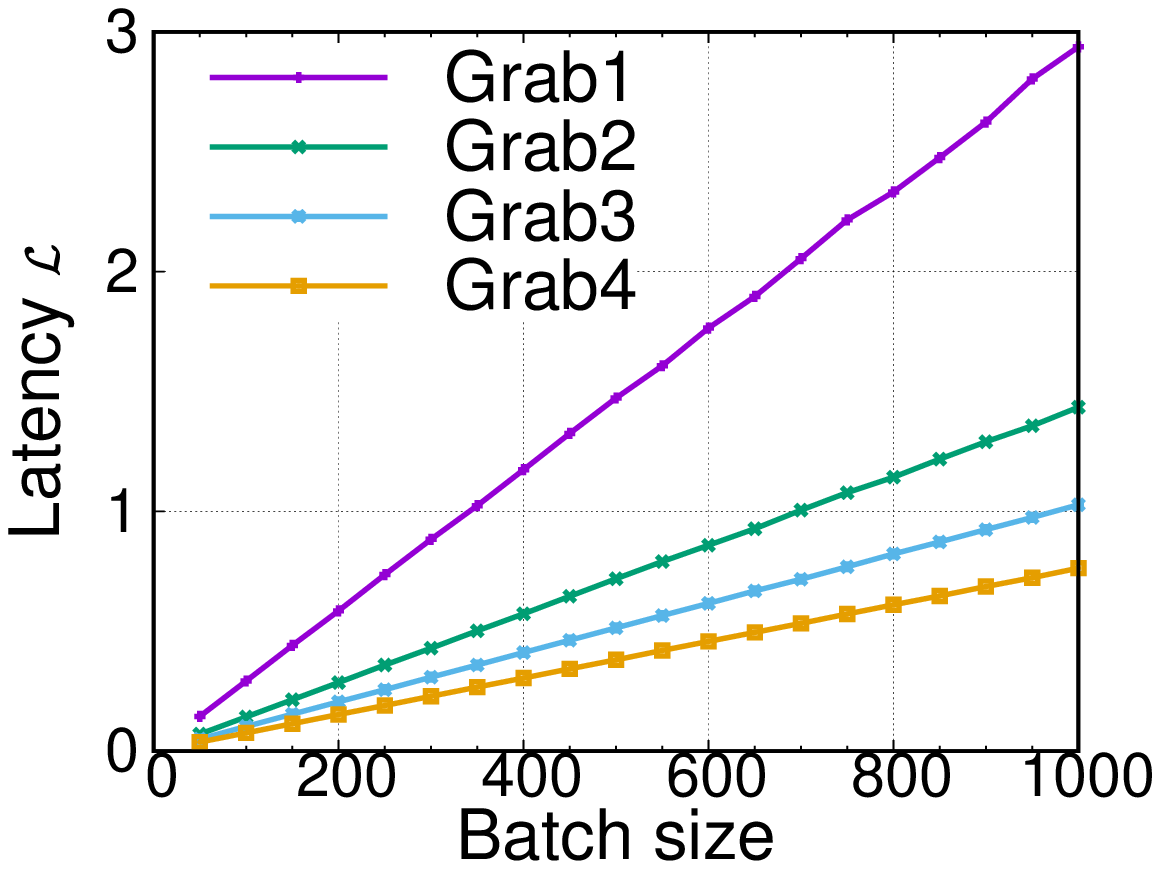}}
\end{minipage}
    \caption{\jiaxin{Elapsed time and latency by varying batch sizes}}\label{fig:tplatency}
\end{figure}

\subsection{Effectiveness of \Spade{}}\label{sec:effectiveness}

\eat{Next, we demonstrate two metrics to measure the effectiveness of \Spade{}. First, fraudsters should be detected as soon as possible to avoid economic loss, \ie the latency $\latency$ of response to fraudulent activities should be reduced (defined by Equation~\ref{eq:latency}).  Second, once we identify a fraudster, we ban his/her following transactions to prevent economic loss. We use $\Ratio$ to denote the ratio of suspicious transactions that are prevented to all suspicious transactions.}

\stitle{Latency.} Our experiment reveals that when the batch size increases, the latency of the batch peeling sequence increases (shown in Figure~\ref{fig:tplatency}). For example, the latency of \IncDENG{} (resp. \IncDENGW{} and \IncFraudar{}) is $0.76$ (resp. $0.74$ and $0.76$). We remarked that $99.99\%$ of the latency of \IncDENG{}, \IncDENGW{} and \IncFraudar{} is the queueing time, \ie \Spade{} accumulates enough transactions and processes them together. Furthermore, the latency in Grab1 is higher than that in Grab4. For example, the latency of \IncFraudar{} in Grab1 (resp. Grab4) is $2.93$ (resp. $0.76$). This is because the queueing time on Grab1 is longer than that on Grab4. 

\stitle{Prevention ratio.} As shown in Figure~\ref{fig:prevent}, the prevention ratio continues to decrease as latency increases on $\Grab$'s datasets. Our results show that \IncDENGU{} (resp. \IncDENGWU{} and \IncFraudarU{}) can prevent $88.34\%$ (resp. $86.53\%$ and $92.47\%$) of fraudulent activities. \IncDENG-$1K$ (resp. \IncDENGW-$1K$ and \IncFraudar-$1K$) can prevent $28.6\%$ (resp. $41.18\%$ and $92.47\%$) of fraudulent activities by excluding queueing time.

\stitle{Case studies.} We next present the effectiveness of \Spade{} in discovering meaningful fraud through case studies in the datasets of $\Grab$. There are three popular fraud patterns as shown in Figure~\ref{fig:casestudy}. First, \textit{customer-merchant collusion} is the customer and the merchant performing fictitious transactions to use the opportunity of promotion activities to earn the bonus (Figure~\ref{fig:casestudy}(a)). Second, there is a group of users who take advantage of promotions or merchant bugs, called \textit{deal-hunter} (Figure~\ref{fig:casestudy}(b)). Third, some merchants recruit fraudsters to create false prosperity by performing fictitious transactions, called \textit{click-farming} (Figure~\ref{fig:casestudy}(c)). All three cases form a dense subgraph in a short period of time.

We investigate the details of the customer-mercant collusion in Figure~\ref{fig:casestudy}(d). \IncDENG{} and $\DENG$ start both at $T_0$. Under the semantic of $\DENG$, the user becomes a fraudster at $T_1$ (one second after $T_0$). \IncDENG{} spots the fraudster at $T_1$ with negligible delay. However, $\DENG$ cannot detect this fraud at $T_1$, as it is still evaluating the graph snapshot at $T_0$. By $\DENG$, this fraudster will be detected after the second round detection of $\DENG$ at $T_2$ (about 60 seconds after $T_0$). During the time period $[T_1,T_2]$, there are $720$ potential fraudulent transactions generated. Similar observations are made in the other two cases. \jiaxin{Due to space limitations, they are presented in Appendix~\ref{sec:casestudy} of \cite{techreport}.}

\begin{figure}[tb]
    \includegraphics[width=0.9\linewidth]{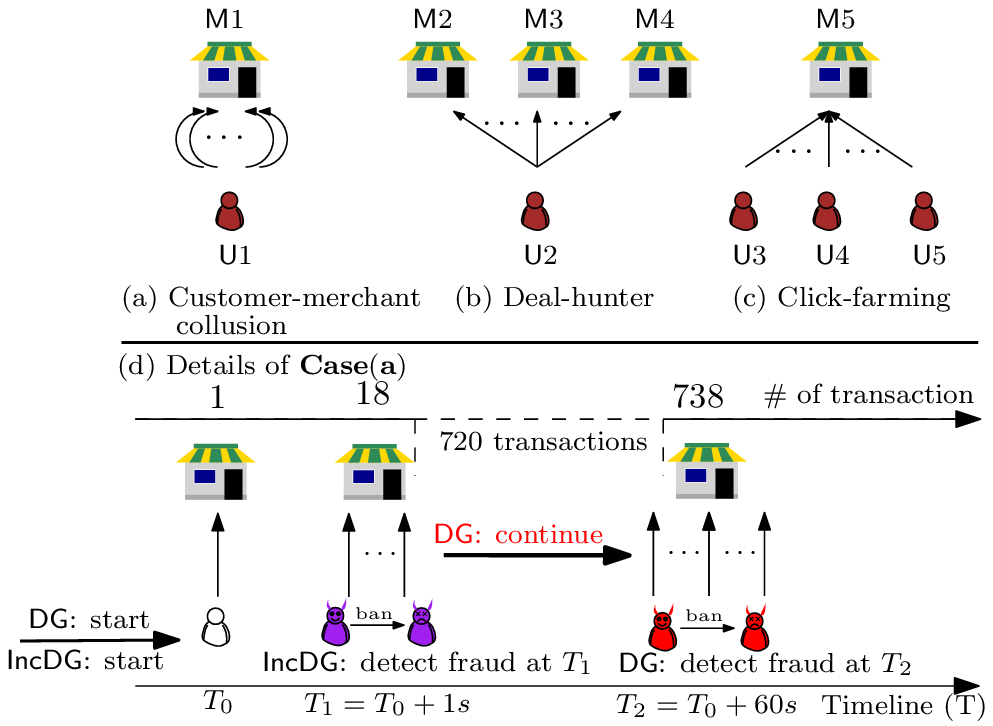}
    \caption{Case study: three fraud patterns}\label{fig:casestudy}
\end{figure}

\section{Related work}\label{sec:related}

\stitle{Dense subgraph mining.} A series of studies have utilized dense subgraph mining to detect fraud, spam, or communities on social networks and review networks~\cite{hooi2016fraudar,shin2016corescope,ren2021ensemfdet}. However, they are proposed for static graphs. Some variants \cite{bahmani2012densest,epasto2015efficient} are designed to detect dense subgraphs in dynamic graphs. \cite{shin2017densealert} is proposed to spot generally dense subtensors created in a short period of time. Unlike these studies, \Spade{} detects the fraudsters on both weighted and unweighted graphs in real time. Moreover, we propose an edge grouping technique which distinguishes potential fraudulent transactions from benign transactions and enables incremental maintenance in batch. 

\stitle{Graph clustering.} A common practice is to employ graph clustering that divides a large graph into smaller partitions for fraud detection. DBSCAN~\cite{gan2015dbscan,ester1996density} and its variant hdbscan ~\cite{mcinnes2017hdbscan} use local search heuristics to detect dense clusters. K-Means~\cite{duda1973pattern} is a clustering method of vector quantization. \cite{yamanishi2004line} detects medical insurance fraud by recognizing outliers. Unlike these studies, \Spade{} is robust with worst-case guarantees in search results. Moreover, \Spade{} provides simple but expressive APIs for developers, which allows their peeling algorithms to be incremental in nature on evolving graphs.

\stitle{Fraud detection using graph techniques.} COPYCATCH~\cite{beutel2013copycatch} and GETTHESCOOP~\cite{jiang2014inferring} use local search heuristics to detect dense subgraphs on bipartite graphs. Label propagation~\cite{wang2015community} is an efficient and effective method of detecting community.  ~\cite{cortes2003computational} explores link analysis to detect fraud. \cite{wang2021deep} and \cite{dou2020enhancing} explore the GNN to detect fraud on the graph. Unlike these studies, \Spade{} detects fraud in real-time and supports evolving graphs.


\section{Conclusion}\label{sec:conclusion}

In this paper, we propose a real-time fraud detection framework called \Spade. We propose three fundamental peeling sequence reordering techniques to avoid detecting fraudulent communities from scratch. \Spade{} enables popular peeling algorithms to be incremental in nature and improves their efficiency. Our experiments show that \Spade{} speeds up fraud detection up to $6$ orders of magnitude and up to $88.34\%$ fraud activities can be prevented.

The results and case studies demonstrate that our algorithm is helpful to address the challenges in real-time fraud detection for the real problems in $\Grab$ but also goes beyond for other graph applications as shown in our datasets.

\begin{acks}
This work was funded by the Grab-NUS AI Lab, a joint collaboration between GrabTaxi Holdings Pte. Ltd. and National University of Singapore. We thank the reviewers for their valuable feedback.
\end{acks}

\newpage

\bibliographystyle{abbrv}
\bibliography{ref}

\newpage
\appendix

\section{Proofs of lemmas}\label{sec:proof}

In this section, we provide all the formal proofs in Section~\ref{sec:Spade} of the main paper.

\begin{manuallemma}{\ref{lemma:seq}} 
$\Seq'[1:i-1] = \Seq[1:i-1]$.
\end{manuallemma}

\begin{proof}
    $\forall k\in [1,i-1]$, $w_{u_i}(S_k)$ and $w_{u_j}(S_k)$ increase by $\Delta$. Therefore, $w_{u_k}(S_k)$ is still the smallest among $S_k$. Hence, $u_k$ will be removed at $k$-th iteration. By induction, $\Seq'[1:i-1] = \Seq[1:i-1]$.
\end{proof}

\begin{lemma}\label{lemma:subset}
    If $S_i \subseteq S_j$ and $u_k\in S_i$, $w_{u_k}(S_j) \geq w_{u_k}(S_i)$.
\end{lemma}

\begin{proof}
By definition, we have the following.
\begin{equation}
\footnotesize
\begin{split}
    w_{u_k}(S_j) & = a_k + \sum_{(u_j\in S_j) \bigwedge ((u_k,u_j)\in E)} c_{kj} + \sum_{(u_j\in S_j) \bigwedge ((u_j,u_k)\in E)} c_{jk} \\
    & = w_{u_k}(S_i) + \sum_{(u_j\in S_j\setminus S_i) \bigwedge ((u_k,u_j)\in E)} c_{kj} + \sum_{(u_j\in S_j\setminus S_i) \bigwedge ((u_j,u_k)\in E)} c_{jk}
\end{split}
\end{equation}
Since the weights on the edges are nonnegative, $w_{u_k}(S_j) > w_{u_k}(S_i)$.
\end{proof}

\begin{manuallemma}{\ref{lemma:peel}}
    If $\Delta_k > \Delta_{\min}$, $u_{\min} = \mathop{\arg\min}\limits_{u\in T\cup S_k}w_{u}(T\cup S_k)$.
\end{manuallemma}

\begin{proof}
    Consider a vertex $u'\in T\cup S_k$, where $u'\not =u_k$ or $u'\not = u_{\min}$. 1) If $u'\in S_k$, due to Lemma~\ref{lemma:subset}, $w_{u'}(T\cup S_k) > w_{u'}(S_k) > w_{u_k}(S_k) \geq w_{u_k}(T\cup S_k) = \Delta_k >\Delta_{\min}$. 2) If $u'\in T$, $w_{u'}(T\cup S_k) > w_{u_{\min}}(T\cup S_k) = \Delta_{\min}$. Hence, $u'$ is not the vertex that has the smallest peeling weight. Therefore, $u_{\min}$ has the smallest peeling weight.
\end{proof}

\begin{lemma}\label{lemma:opt}
    If $\exists u\in S$, such that $w_{u}(S) < g(S^*)$, then $S\not = S^{*}$. 
\end{lemma}

\begin{proof}
    We prove it in contradiction by assuming that $S = S^{*}$. By peeling $u$ from $S$, we have the following.
    \begin{equation}
    \footnotesize
        \begin{split}
            g(S^{*}\setminus \{u\}) & = \frac{f(S^{*}) - w_{u}(S^{*})}{|S^{*}|-1} > \frac{f(S^{*}) - g(S)}{|S^{*}|-1} \\
            & = \frac{f(S^{*}) - g(S^*)}{|S^{*}|-1} = \frac{f(S^{*}) - \frac{f(S^{*})}{|S^{*}|}}{|S^{*}|-1} = g(S^*)
        \end{split}       
    \end{equation}
A better solution can be obtained by peeling $u_i$ from $S^*$. This contradicts the notion that $S^*$ is the optimal solution. Hence, $S_i\not = S^{*}$.
\end{proof}

\begin{manuallemma}{\ref{lemma:optimalset}}
Given an edge $e = (u_i,u_j)$, if $e$ is a benign edge, after the insertion of $e$, $u_i\not \in S^{*}$ and $u_j\not \in S^{*}$.
\end{manuallemma}

\begin{proof}
    We prove this lemma in contradiction by assuming that $u_i\in S^{*}$. $w_{u_i}(S^*) \leq w_{u_i}(S_0) + c_{ij} < g(S^P) \leq g(S^*)$. We have $S^*\not = S^*$ due to Lemma~\ref{lemma:opt}. We can conclude that $u_i\not\in S^*$. Similarly, $u_j\not \in S^*$.
\end{proof}

\begin{lemma}\label{lemma:opt2}
If $\exists u\in S_i$, $w_{u}(S_i) < g(S_i)$, then $S_i\not = S^{P}$.
\end{lemma}

\begin{proof}
We prove this in contradiction by assuming that $S_i= S^P$. Suppose that $u_i$ is peeled from $S_i$. Hence, $w_{u_i}(S^P) \leq w_{u}(S^P)$ due to the peeling definition. The proof can be obtained as follows:
\begin{equation}
\footnotesize
    \begin{split}
        g(S^P \setminus \{u_i\}) & = \frac{f(S^{P}) - w_{u_i}(S^{P})}{|S^{P}|-1} > \frac{f(S^P) - w_{u}(S^P)}{|S^P|-1} > g(S^P)   
    \end{split}
\end{equation}
This contradicts the fact that $S^P$ has the highest density. We can conclude that $S_i \not = S^{P}$.
\end{proof}

\begin{manuallemma}{\ref{lemma:benign}}
Given a benign edge $e = (u_i,u_j)$ insertion, at least one of the following two conditions is established: 1) $u_i\not \in S^{P'}$ and $u_j\not \in S^{P'}$; and 2) $g(S^{P'}) < g(S^P)$.
\end{manuallemma}

\begin{proof}
    Without loss of generality, we assume $i\leq j$. We prove this in contradiction by assuming $g(S^{P'}) \geq g(S^P)$ and $u_i\in S^{P'}$ or $u_j\in S^{P'}$ after inserting the edge $e$.

    Due to Lemma~\ref{lemma:subset} and $S^{P'} \subseteq S_0$, we have

    \begin{equation}\label{eq:subset}
        w_{u_i}(S^{P'}) \leq w_{u_i}(S_0) < w_{u_i}(S_0) + c_{ij} < g(S^P) < g(S^{P'})
    \end{equation}

    Therefore, $S^{P'}$ is not the result returned by peeling algorithms due to Lemma~\ref{lemma:opt2} which contradicts that $S^{P'}$ maximizes $g$.

\end{proof}

\section{More case studies}\label{sec:casestudy}

\begin{figure}[tb]
    \includegraphics[width=\linewidth]{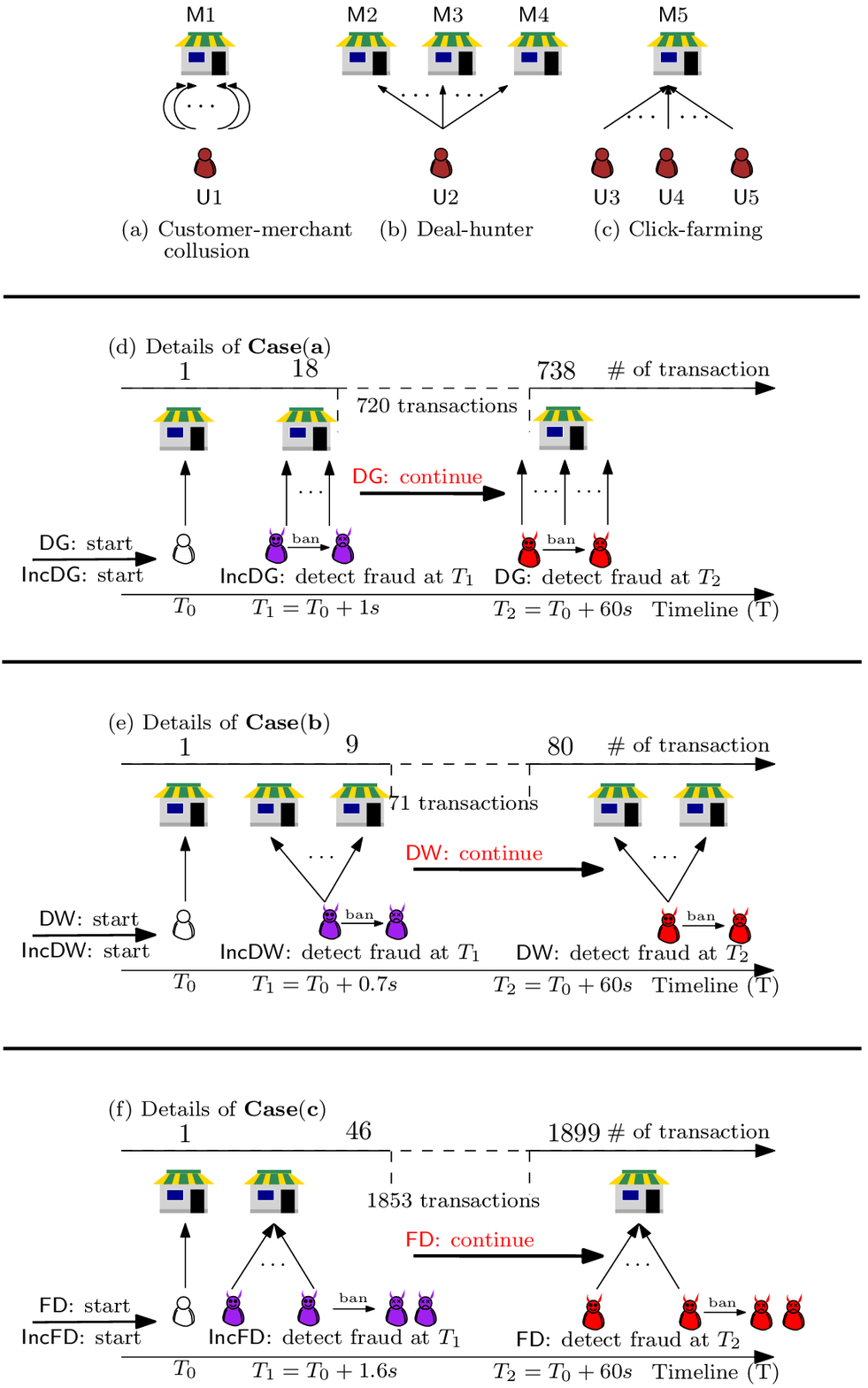}
    \caption{Case study: three fraud patterns}\label{fig:casestudyappendix}
\end{figure}

We next present the effectiveness of \Spade{} in discovering meaningful fraud through case studies in the datasets of $\Grab$. There are three popular fraud patterns as shown in Figure~\ref{fig:casestudyappendix}. First, \textit{customer-merchant collusion} is the customer and the merchant performing fictitious transactions to use the opportunity of promotion activities to earn the bonus (Figure~\ref{fig:casestudyappendix}(a)). Second, there is a group of users who take advantage of promotions or merchant bugs, called \textit{deal-hunter} (Figure~\ref{fig:casestudyappendix}(b)). Third, some merchants recruit fraudsters to create false prosperity by performing fictitious transactions, called \textit{click-farming} (Figure~\ref{fig:casestudyappendix}(c)). All three cases form a dense subgraph in a short period of time.

\etitle{Customer-merchant collusion.} We detail the customer-mercant collusion in Figure~\ref{fig:casestudyappendix}(d). \IncDENG{} and $\DENG$ start both at $T_0$. Under the semantic of $\DENG$, the user becomes a fraudster at $T_1$ (one second after $T_0$). \IncDENG{} spots the fraudster at $T_1$ with negligible delay. However, $\DENG$ cannot detect this fraud at $T_1$, as it is still evaluating the graph snapshot at $T_0$. By $\DENG$, this fraudster will be detected after the second round detection of $\DENG$ at $T_2$ (about 60 seconds after $T_0$). During the time period $[T_1,T_2]$, there are $720$ potential fraudulent transactions generated.

\etitle{Deal-hunter.} We investigate the details of customer-merchant collusion in Figure~\ref{fig:casestudyappendix}(e). \IncDENGW{} and $\DENGW$ start both at $T_0$. Under the semantic of $\DENGW$, the user becomes a fraudster at $T_1$ ($0.7$ second after $T_0$). \IncDENGW{} identifies the fraudster at $T_1$ with negligible delay. However, $\DENGW$ cannot detect this fraud at $T_1$, as it is still evaluating the graph snapshot at $T_0$. By $\DENGW$, this fraudster will be detected after the second round detection of $\DENGW$ at $T_2$ (about 60 seconds after $T_0$). During the time period $[T_1,T_2]$, there are $71$ potential fraudulent transactions generated.

\etitle{Click-farming.} Last but not least, we present the details of click-farming in Figure~\ref{fig:casestudyappendix}(f).  \IncFraudar{} and $\Fraudar$ start both at $T_0$. Under the semantic of $\Fraudar$, the group of users becomes fraudsters at $T_1$ ($1.6$ second after $T_0$). \IncFraudar{} spots the fraudsters at $T_1$ with negligible delay. However, $\Fraudar$ cannot detect this fraud at $T_1$, as it is still evaluating the graph snapshot at $T_0$. By $\Fraudar$, these fraudsters will be detected after the second round detection of $\Fraudar$ at $T_2$ (about 60 seconds after $T_0$). During the time period $[T_1,T_2]$, there are $1853$ potential fraudulent transactions generated.

\begin{figure}[tb]
    \includegraphics[width=0.65\linewidth]{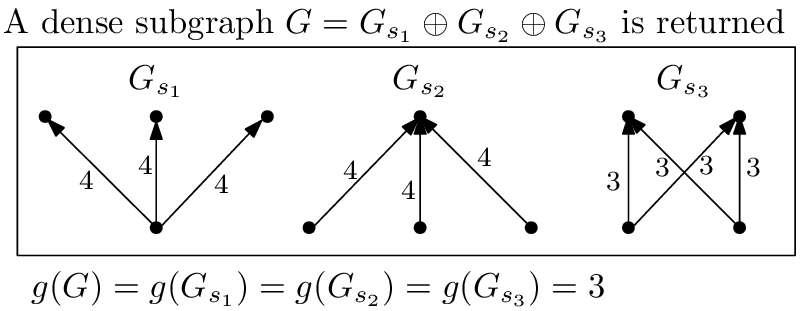}
    \caption{Multiple fraud instances}\label{fig:dense}
\end{figure}

\begin{figure*}[tb]
\includegraphics[width=1\linewidth]{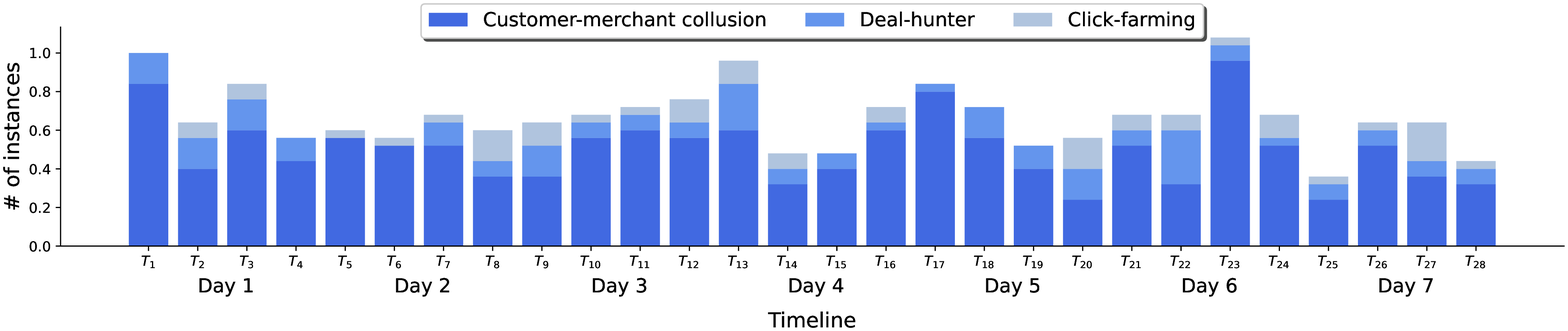}
    \caption{\jiaxin{\Spade{} spots and enumerates the new fraudsters. The appearances of dense subgraphs indicates various types frauds including customer-merchant collusion, deal-hunter and click-farming. We show that fraudulent instances are identified in a week. Each bar represents the number of fraudulent instances are detected in the corresponding timespan. The numbers are normalized to the number of fraudulent instances during the first timespan.}}\label{fig:spot}
\end{figure*}

\noindent\jiaxin{Consider a dense subgraph $G$, it could consists of multiple fraud instances as shown in Figure~\ref{fig:dense}. $G$ consists of $G_{s_1}$, $G_{s_2}$ and $G_{s_3}$ and all of their densities are equal to $3$. Therefore, all will be returned, since they commonly form a dense subgraph $G$. We enumerate these instances once new fraudsters are identified.}

\noindent\jiaxin{\stitle{Fraud enumeration.} Figure~\ref{fig:spot} depicts the new fraudsters identified by \Spade{} in $28$ timespans. Once new fraudsters are detected, \Spade{} enumerates them and reports them to the moderators. In Figure~\ref{fig:spot}, each bar represents the number of fraudulent instances are detected in the corresponding timespan. We investigated the detected fraudsters and found that most of their transactions corresponded to actual fraud, including customer-merchant collusion, deal-hunter and click-farming.}

\eat{, $\mathsf{M}70980$ is highly connected to several fraudsters, such as $\mathsf{U}30592$, $\mathsf{U}7036$, $\mathsf{U}90618$, etc. With \Spade{}, \IncFraudar{} detects $\mathsf{M}70980$ in the $46$th transaction, while $\Fraudar$ detects him/her in the $1899$th transaction.}

\section{Future extensions}\label{sec:extension}

We discuss a few possible extensions of our current system, including edge deletion, enumeration and fraud detection within a given period of time. 

\subsection{Peeling sequence reordering with edge deletion}\label{sec:extension:deletion}

The company will delete some outdated transactions since they are not of much value for fraud detection in some operational demands, \eg some transactions generated several years ago. Given such an operational demand, we consider the extension of incremental maintenance with edge deletion of $(u_i,u_j)$ (without loss of generality, we assume $i < j$). A straightforward solution is also to reorder the peeling sequence. We summarize the key steps as follows and leave the extension details of \Spade{} in future work.

\stitle{Incremental algorithm (\TCald).} \TCald{} initializes an empty vector for the updated peeling sequence $\Seq'$. \Spade{} maintains a pending queue $T$ to store the vertices pending reordering. We iteratively compare 1) the head of $T$, denoted by $u_{\min}$ and 2) the vertex $u_k$ in the peeling sequence $\Seq$, where $k < i$. The corresponding peeling weights are denoted by $\Delta_{\min}$ and $\Delta_k$. We consider the following two cases.

\stitle{Case 1.} If the peeling weight $w_{u_k}(S_0) > \Delta_{\min}$, we insert $u_k$ into $T$ and update the priorities in $T$ for the neighbors of $u_k$, $N(u_{\max})$, $k=k-1$.

\stitle{Case 2.} If the peeling weight $w_{u_k}(S_0) \leq \Delta_{\min}$, we append $\Seq[1:k]$ to $\Seq'[1:k]$.

\begin{figure*}[tb]
\includegraphics[width=0.85\linewidth]{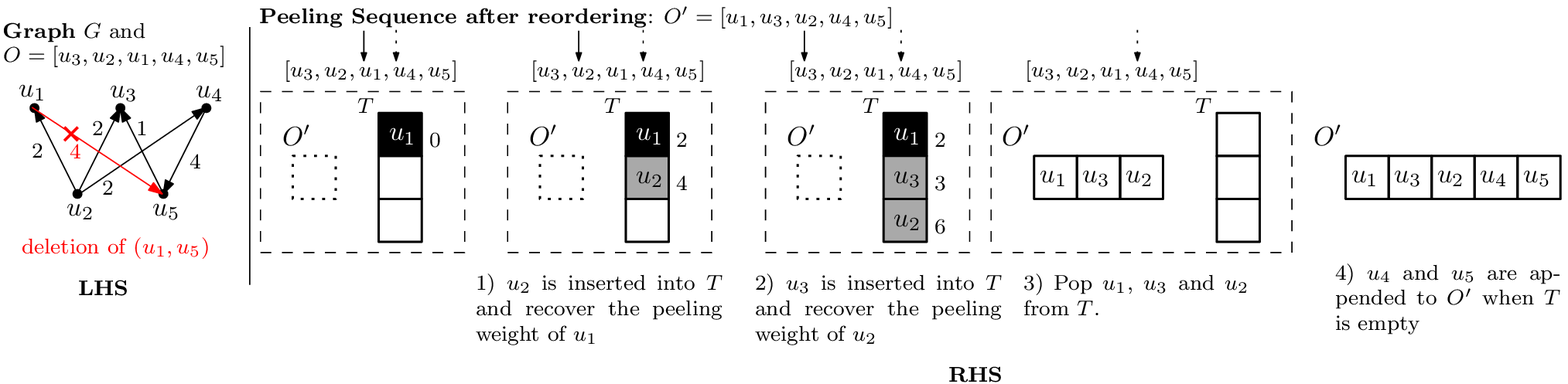}
    \caption{Peeling sequence reordering  with edge deletion (A running example)}\label{fig:run:deleted}
\end{figure*}

While $T$ is non-empty, we iteratively compare 1) the head of $T$ and 2) the vertex $u_k$ in the peeling sequence $\Seq$, where $k \geq i + 1$. The incremental maintenance is identical to that of edge insertion in Section~\ref{sec:edgebyedge}. Specifically, we consider the following three cases:

\stitle{Case 1.} If $\Delta_{\min} < \Delta_k $, we pop the $u_{\min}$ from $T$ and insert it to $\Seq'$. Then we update the priorities in $T$ for the neighbors of $u_{\min}$, $N(u_{\min})$.

\stitle{Case 2.} If $\Delta_{\min} \geq \Delta_{k}$ and $\exists u_{T}\in T, (u_{T},u_{k})\in E$ or $(u_{k},u_{T})\in E$, we insert $u_k$ into $T$. The peeling weight is $w_{u_k}(T\cup S_k)$ $ = \Delta_k $ $ + $ $\sum_{(u_T\in T) \bigwedge ((u_T,u_k)\in E)}$ $ c_{Tk} + $ $\sum_{(u_T\in T) \bigwedge ((u_k,u_T)\in E)} c_{kT}$, $k=k+1$.

\stitle{Case 3.} If $\Delta_{\min} \geq \Delta_{k}$ and $\forall u_{T}\in T, (u_{T},u_{k})\not\in E$ and $(u_{k},u_{T})\not\in E$, we insert $u_{k}$ to $\Seq'$, $k=k+1$.

\setlist{nolistsep}
\eat{
\begin{itemize}[leftmargin=*]
    \item \stitle{Case 1.} If $\Delta_{\min} < \Delta_k $, we pop the $u_{\min}$ from $T$ and insert it to $\Seq'$. Then we update the priorities in $T$ for the neighbors of $u_{\min}$, $N(u_{\min})$.
    \item \stitle{Case 2.} If $\Delta_{\min} \geq \Delta_{k}$ and $\exists u_{T}\in T, (u_{T},u_{k})\in E$ or $(u_{k},u_{T})\in E$, we insert $u_k$ into $T$. The peeling weight is $w_{u_k}(T\cup S_k)$ $ = \Delta_k $ $ + $ $\sum_{(u_T\in T) \bigwedge ((u_T,u_k)\in E)}$ $ c_{Tk} + $ $\sum_{(u_T\in T) \bigwedge ((u_k,u_T)\in E)} c_{kT}$, $k=k+1$.
    \item \stitle{Case 3.} If $\Delta_{\min} \geq \Delta_{k}$ and $\forall u_{T}\in T, (u_{T},u_{k})\not\in E$ and $(u_{k},u_{T})\not\in E$, we insert $u_{k}$ to $\Seq'$, $k=k+1$.
\end{itemize}
}

We repeat the above iteration until $T$ is empty.

\begin{example}
    Consider the graph $G$ in Figure~\ref{fig:run:deleted} and its peeling sequence $\Seq = [u_3, u_2, u_1, u_4, u_5]$. Suppose that an outdated edge $(u_1,u_5)$ is deleted from $G$ as shown in the LHS of Figure~\ref{fig:run:deleted}. The reordering procedure is presented in the RHS of Figure~\ref{fig:run:deleted}. $u_1$ is pushed to the pending queue $T$. Since the peeling weights $w_{u_2}(S_0)$ and $w_{u_3}(S_0)$ are larger than the peeling weight of $u_1$. $u_2$ and $u_3$ are inserted into $T$. Since the peeling weight of $u_1$ is less than that of $u_4$, it will be appended to $\Seq'$. Similarly $u_3$ and $u_2$ are appended to $\Seq'$ accordingly. Once $T$ is empty, the rest of the vertices, $u_4$ and $u_5$, in $\Seq$ are appended to $\Seq'$ directly. Therefore, the reordered peeling sequence is $\Seq'=[u_1,u_3,u_2,u_4,u_5]$.
\end{example}

\subsection{Dense subgraph enumeration}

In case of the enumeration of dense subgraphs due to some operational demands, we consider both static graphs and dynamic graphs.

\stitle{Static graphs.} Given a graph $G=(V,E)$, peeling algorithm $Q$ returns $S^P$. To enumerate dense subgraphs, we can perform the peeling algorithm $Q$ by removing $S^P$ from $G$, denoted by $G'=(V', E')$. Specifically, $V' = V\setminus S^P$ and $E' = E\setminus E^P$, where $\forall (u_i,u_j) \in E^P, u_i\in S^P$ or $u_j\in S^P$. Therefore, ${S^P}'$ will be returned as the second densest subgraph. We can perform the peeling algorithm $Q$ recursively to enumerate all dense subgraphs.

It is remarkable that we do not have to compute ${S^P}'$ from scratch. Instead, we can perform the incremental maintenance of edge deletion as introduced in Section~\ref{sec:extension:deletion}. 

\stitle{Dynamic graphs.} Given a graph $G$ and graph updates $\Delta G = (\Delta V, \Delta E)$, a straightforward solution is to reorder the peeling sequence by Algorithm~\ref{algo:batch} first. For the enumeration, we can think of this dynamic graph $G \oplus \Delta G$ as a static graph.

\subsection{Fraud detection during some time period}

Given a graph $G=(V,E)$ generated during a timespan $[\tau_s, \tau_e]$ ($\tau_s < \tau_e$) and the peeling sequence $\Seq = Q(G)$. Taking a new graph $G'=(V',E')$ generated during a timespan $[\tau_{s'},\tau_{e'}]$, we would like to identify the peeling sequence on $G'$, \ie $\Seq' = Q(G')$. To simply our discussion, we denote a set of edges generated during timespan $[\tau_s, \tau_e]$ by $E_{[s,e]}$

\stitle{Case 1.} If $\tau_{e'} < \tau_s$ or $\tau_e < \tau_{s'}$, $G$ and $G'$ do not overlap. Therefore, we directly apply the peeling algorithm $Q$ on $G'$.

\stitle{Case 2.} If $\tau_{s'} < \tau_s$ and $\tau_e < \tau_{e'}$, we perform Algorithm~\ref{algo:batch} by inserting two sets of edges, $E_{[s',s]}$ and $E_{[e,e']}$ to $G$. Then we can identify the peeling sequence $\Seq'$ on $G'$.

\stitle{Case 3.} If $\tau_s < \tau_{s'}$ and $\tau_{e'} < \tau_e$, we perform incremental maintenance in Section~\ref{sec:extension:deletion} by deleting two sets of edges, $E_{[s,s']}$ and $E_{[e',e]}$ from $G$. Then we can identify the peeling sequence $\Seq'$ on $G'$.

\stitle{Case 4.} If $\tau_{s'} < \tau_{s} < \tau_{e'} < \tau_e$, we perform Algorithm~\ref{algo:batch} by inserting a set of edges, $E_{[s',s]}$ to $G$ and perform incremental maintenance in Section~\ref{sec:extension:deletion} by deleting a set of edges $E_{[e',e]}$ from $G$.

\stitle{Case 5.} If $\tau_{s} < \tau_{s'} < \tau_{e} < \tau_{e'}$, we perform Algorithm~\ref{algo:batch} by inserting a set of edges, $E_{[e,e']}$ to $G$ and perform incremental maintenance in Section~\ref{sec:extension:deletion} by deleting a set of edges $E_{[s,s']}$ from $G$.

\begin{figure}[tb]
    \includegraphics[width=0.95\linewidth]{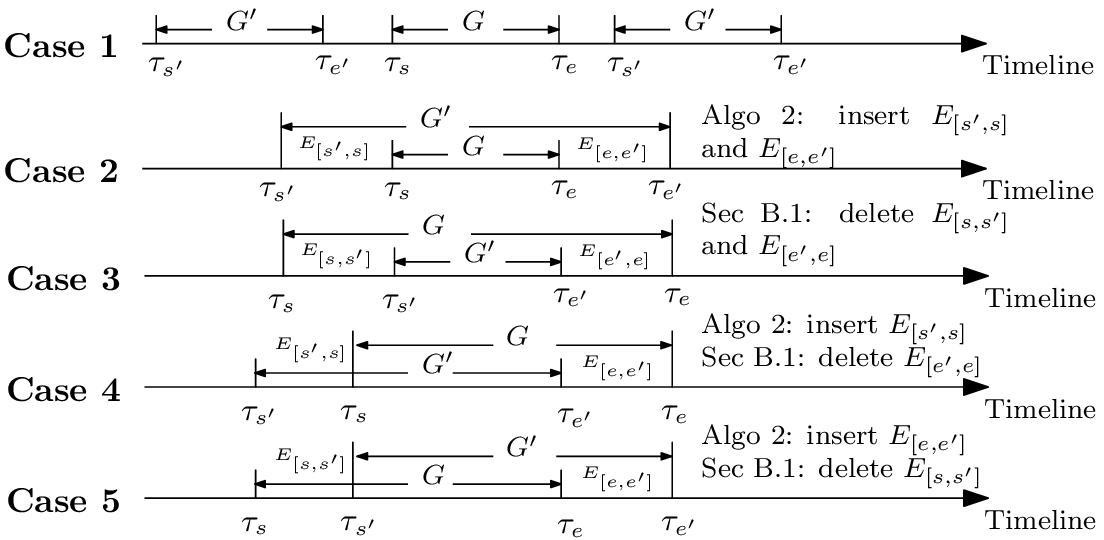}
    \caption{Fraud detection during some time period}\label{fig:timeperiod}
\end{figure}

\section{Accuracy guarantee of Algorithm~\ref{algo:batch}}\label{sec:correctness}

\stitle{Correctness and accuracy guarantee.} In \textbf{Case 1}, if $\Delta_k > \Delta_{\min}$, $u_{\min}$ is chosen to insert to $\Seq'$ since it has the smallest peeling weight due to Lemma~\ref{lemma:peel}. In \textbf{Case 2(b)}, $\Delta_k$ is the smallest peeling weight and $u_k$ is chosen to insert to $\Seq'$. The peeling sequence is identical to that of $G\oplus \Delta G$, since in each iteration the vertex with the smallest peeling weight is chosen. The accuracy of the worst-case is preserved due to Lemma~\ref{lemma:2ppr}.

\section{Properties of density metrics}\label{sec:axioms}

\newtheorem{axiom}{Axiom}
\stitle{Density metrics $g$.} We adopt the class of metrics $g$ in previous studies~\cite{hooi2016fraudar,gudapati2021search,charikar2000greedy}, $g(S) = \frac{f(S)}{|S|}$, where $f$ is the total weight of $G[S]$, \ie the sum of the weight of $S$ and $E[S]$:

\begin{equation}\label{eq:density}
    f(S)=\sum_{u_i\in S} a_i + \sum_{u_i,u_j\in S \bigwedge (u_i,u_j)\in E} c_{ij}
\end{equation}

We use $f_E(S)$ to denote the total suspiciousness of the edges $E[S]$ and $f_V(S)$ to denote the total suspiciousness of $S$, \ie

\begin{equation}
    f_V(S)= \sum_{u_i\in S} a_i 
\end{equation}

and 

\begin{equation}
    f_E(S)= \sum_{u_i,u_j\in S \bigwedge (u_i,u_j)\in E} c_{ij}
\end{equation}

The density metric defined in Equation~\ref{eq:density} satisfies Axiom~\ref{ax:first}-\ref{ax:fourth}. We adapted these basic properties from \cite{jiang2015general}.

\begin{axiom}
\label{ax:first}
  \stitle{[Vertex suspiciousness]} If 1) $|S| = |S'|$, 2) $f_E(S) = f_E(S')$, and 3) $f_V(S) > f_V(S')$, then $g(S) > g(S')$.
\end{axiom}

\begin{proof}
\begin{equation}
    g(S) = \frac{f_V(S) + f_E(S)}{|S|} > \frac{f_V(S') + f_E(S')}{|S'|} = g(S') 
\end{equation}    
\end{proof}

With slight abuse of definition, we use $g(S(V,E))$ to denote the total suspiciousness of $S$ on the graph $G=(V,E)$.

\begin{axiom}
\label{ax:second}
  \stitle{[Edge suspiciousness]} If $e=(u_i,u_j)\not\in E$, then $g(S(V,E\cup \{e\})) > g(S(V,E))$.
\end{axiom}
\begin{proof}
    \begin{equation}
        g(S(V,E\cup \{e\})) = \frac{f_V(S) + f_E(S) + c_{ij}}{|S|} > \frac{f_V(S) + f_E(S)}{|S|} = g(S)  
    \end{equation}
\end{proof}

\begin{axiom}
\label{ax:fourth}
  \stitle{[Concentration]} If $|S| < |S'|$ and $f(S) = f(S')$, then $g(S)>g(S')$.
\end{axiom}
\begin{proof}
    \begin{equation}
        g(S) = \frac{f(S)}{|S|} > \frac{f(S')}{|S'|} = g(S')
    \end{equation}
\end{proof}


\section{Instances of \Spade{}}\label{sec:instances}

We show that the popular peeling algorithms can be easily implemented and supported by \Spade{}, \eg $\DENG$~\cite{charikar2000greedy}, $\DENGW{}$~\cite{gudapati2021search} and $\Fraudar{}$~\cite{hooi2016fraudar}.

\etitle{Instance 1. Dense subgraphs ($\DENG{}$)~\cite{charikar2000greedy}.} $\DENG{}$ is designed to quantify the connectivity of substructures. It is widely used to detect fake comments~\cite{kumar2018community} and fraudulent activities~\cite{ban2018badlink} on social graphs. Let $S\subseteq V$. The density metric of $\DENG{}$ is defined by $g(S) = \frac{|E[S]|}{|S|}$. To implement $\DENG$ on \Spade{}, developers only need to design and plug in the suspiciousness function $\mathsf{esusp}$ by calling $\mathsf{ESusp}$. Specifically,  $\mathsf{esusp}$ is a constant function for edges, \ie $\mathsf{esusp}(u_i,u_j) = 1$.

\etitle{Instance 2. Dense weighted subgraphs ($\DENGW$)~\cite{gudapati2021search}.} On transaction graphs, there are weights on the edges in usual, such as the transaction amount. The density metric of $\DENGW{}$ is defined by $g(S) = \frac{\sum_{(u_i,u_j)\in E[S]}c_{ij}}{|S|}$, where $c_{ij}$ is the weight of the edge $(u_i,u_j)\in E$. To implement $\DENGW{}$, users only need to plug in the suspiciousness function $\mathsf{esusp}$, \ie given an edge, $\mathsf{esusp}(u_i,u_j) = c_{ij}$.

\etitle{Instance 3. Fraudar ($\Fraudar{}$)~\cite{hooi2016fraudar}.} To resist the camouflage of fraudsters, Hooi et al. ~\cite{hooi2016fraudar} proposed $\Fraudar{}$ to weight edges and set the prior suspiciousness of each vertex with side information. Let $S\subseteq V$. The density metric of $\Fraudar$ is defined as follows:

\begin{equation}
    g(S) = \frac{f(S)}{|S|} = \frac{\sum_{u_i\in S} a_i + \sum_{u_i,u_j\in S \bigwedge (u_i,u_j)\in E} c_{i,j}}{|S|}
\end{equation}

\begin{center}
  \scriptsize
\begin{lstlisting}[caption=Implementation of $\Fraudar{}$ on \Spade, label={lst:listing-fd}, language=C++]
double vsusp(Vertex v, Graph g){
  return g.weight[v]; //side information on vertex
}
double esusp(Edge e, Graph g){
  return 1/log(g.deg[e.src]+5); //user-defined function
}
int main() {
  Spade spade;  
  spade.VSusp(vsusp); //plug in vsusp (line 1-3) 
  spade.ESusp(esusp); //plug in esusp (line 4-6)
  spade.TurnOnEdgeGrouping(); //enable edge grouping
  spade.LoadGraph("graph_sample_path");
  vector<Vertex> fraudsters = spade.Detect();
  //edge insertions prepared by developers
  vector<Edge> edge_insertions;
  for(Edge e: edge_insertions){
      fraudsters = spade.InsertEdge(e);
  }
  return 0;
}
\end{lstlisting}
\end{center}

To implement $\Fraudar$ on \Spade{}, users only need to plug in the suspiciousness function $\mathsf{vsusp}$ for the vertices by calling $\mathsf{VSusp}$ and the suspiciousness function $\mathsf{esusp}$ for the edges by calling $\mathsf{ESusp}$. Specifically, 1) $\mathsf{vsusp}$ is a constant function, \ie given a vertex $u$, $\mathsf{vsusp}(u) = a_i$ and 2) $\mathsf{esusp}$ is a logarithmic function such that given an edge $(u_i,u_j)$, $\mathsf{esusp}(u_i,u_j) = \frac{1}{\log (x+c)}$, where $x$ is the degree of the object vertex between $u_i$ and $u_j$, and $c$ is a small positive constant ~\cite{hooi2016fraudar}.

Developers can easily implement customized peeling algorithms with \Spade{}, which significantly reduces the engineering effort. For example, users write only about $20$ lines of code (compared to about $100$ lines in the original $\Fraudar$~\cite{hooi2016fraudar}) to implement $\Fraudar$ as shown in List~\ref{lst:listing-fd}. \Spade{} enables $\Fraudar{}$ to be incrmental by nature. Similar observations are made in $\DENG$ and $\DENGW$.

\end{document}